\documentclass[11pt]{article}
\usepackage{amsfonts}
\usepackage{amsmath, amsfonts, amsthm, amssymb, amscd, mathrsfs}
\usepackage[mathcal]{euscript}
\usepackage{color}
\usepackage{lineno, blindtext}
\usepackage {a4wide}
\usepackage{mathtools}
\usepackage[unicode=true,pdfusetitle,
bookmarks=true,bookmarksnumbered=false,bookmarksopen=false,
breaklinks=false,pdfborder={0 0 0},pdfborderstyle={},
backref=false,colorlinks=false]
{hyperref}

\theoremstyle{plain}
\newtheorem{thm}{\protect\theoremname}
\theoremstyle{definition}

\theoremstyle{plain}
\newtheorem{lem}{\protect\lemmaname}
\newtheorem{remark}{\protect\remarkname}
\theoremstyle{plain}
\newtheorem{prop}{\protect\propositionname}

\providecommand{\definitionname}{Definition}
\providecommand{\lemmaname}{Lemma}
\providecommand{\theoremname}{Theorem}
\providecommand{\remarkname}{Remark}
\providecommand{\propositionname}{Proposition}

\begin{document}
\title{Local well-posedness of the initial value problem in Einstein-Cartan theory}

\author{Paulo Luz
\footnote{Centro de Astrofísica e Gravitação - CENTRA,  Instituto Superior Técnico, Universidade de Lisboa, Av. Rovisco Pais 1, 1049-001 Lisboa, Portugal.
Email: paulo.luz@tecnico.ulisboa.pt
\newline
$^\dag$ Centre for Mathematical Analysis, Geometry and Dynamical Systems, Instituto Superior T\'ecnico,  Universidade de Lisboa, Av. Rovisco Pais, 1049–001 Lisbon, Portugal. Email: filipecmena@tecnico.ulisboa.pt
\newline
$\flat$ Centre of Mathematics, University of Minho, 4710-057 Braga, Portugal.
}~ and Filipe C. Mena$^{\dag\flat}$}

\date{\today}
\maketitle

\begin{abstract}
We study the initial value problem in Einstein-Cartan theory which includes torsion and, therefore, a non-symmetric connection on the spacetime manifold. Generalizing the path of a classical theorem by Choquet-Bruhat and York for the Einstein equations, we use a $n+1$ splitting of the manifold and compute the evolution and constraint equations for the Einstein-Cartan system. In the process, we derive the Gauss-Codazzi-Ricci equations including torsion. We prove that the constraint equations are preserved during evolution. Imposing a generalized harmonic gauge, it is shown that the evolution equations can be cast as a quasilinear system in a Cauchy regular form with a characteristic determinant having a non-trivial multiplicity of characteristics.
Using the Leray-Ohya theory for non-strictly hyperbolic systems we then establish the local geometric well-posedness of the Cauchy problem,  for sufficiently regular initial data. For vanishing torsion we recover the classical results for the Einstein equations.
\end{abstract}
\newpage
\tableofcontents
\newpage
\section{Introduction}

{\bf Motivation}. There has been an increasing interest in generalized theories of gravity in astrophysical and cosmology  \cite{Boehmer} while observations are beginning to reach a level of precision that can be used to test some of them \cite{Ferreira}.
However most generalized gravity theories are not yet built on solid mathematical grounds. In fact several theories lead to instabilities  \cite{Felice} and might not be suitable alternative models. 
Furthermore, the well-posedness of the initial value problem has rarely been analyzed and
it has been shown that 
some theories 
with higher order curvature terms can lead to ill-defined Cauchy problems \cite{LeFlochMa17a}.
Given the significant amount of scientific literature considering generalized theories and their abundant applications (see e.g. \cite{yunes-siemens}), it is crucial to establish rigorous results particularly considering the well-posedness of the initial value problem for the field equations. This  not only puts in more solid mathematical grounds some heuristic and numerical results, but  could also trigger new numerical studies as, in part, they have been hampered by the fact that analytical results are lacking to properly implement the numerics.
\\\\
{\bf Past results}. The local well-posedness of the Cauchy problem for vacuum Einstein equations established by Choquet-Bruhat \cite{Bruhat} constitutes a cornerstone of General Relativity. Subsequent developments in the theory of semi-linear hyperbolic PDEs allowed generalisations by weakening the differentiability requirements in the proof, thus allowing more general initial data sets \cite{Marsden, Hughes}. These pioneering works used the harmonic gauge. The Cauchy problem in modified harmonic gauge was studied in \cite{Choquet-Bruhat_York_1995,Choquet-Bruhat_Ruggeri_1983} who obtained hyperbolic equations for more general matter fields (for a review, see \cite{Choquet-Bruhat_Book_2009}). Some cases of dissipative fluids led to non-strictly hyperbolic formulations for which well-posedness can still be obtained \cite{Choquet-Bruhat_Book_2009}. The underlying theory for those cases was developed by Leray and Ohya which considered function spaces with a regularity between $C^\infty$ and analiticity, in the so-called Gevrey classes \cite{Leray-Ohya}.
Generalizations of the results in GR to modified theories of gravity which obtain hyperbolic formulations and show local well-posedness include some interesting cases of scalar-tensor theories \cite{Avalos-JMP, Cooke, Noakes, Salgado, Tey} and $f(R)$ theory \cite{LeFlochMa17a} with particular matter fields.
\\\\
{\bf Numerical Relativity:} Amongst the Numerical Relativity community there have been important efforts to find strongly hyperbolic formulations of generalized fields equations. Notably it has been found that, in some physically relevant cases, the equations of motion in Lovelock and Horndeski theories are not strongly hyperbolic around weak field solutions \cite{papallo-reall, Reall14,Ripley}. In fact, in those theories there are cases where there is lack of uniqueness of solutions and change of character in the evolution equations which can change the causal structure \cite{Bernard, Delsate}. 
An important aspect resulting from these works is the dependence of hyperbolicity on the gauge choice. For instance, although Horndeski theory is not strongly hyperbolic in the harmonic gauge \cite{papallo-reall}, it is so in a modified harmonic gauge \cite{Kovacs}. Interestingly, strong hyperbolicity has also been shown for some models of scalar-Gauss-Bonnet gravity \cite{Kovacs1}, triggering further developments and applications of those theories (see e.g. \cite{East,Bezares}).
\\\\
{\bf Einstein-Cartan theory:}
A natural extension of general relativity is the Einstein-Cartan theory where the underlying geometry is the Riemann-Cartan geometry, characterized by the metric tensor and the affine connection.
One of the interests in the Einstein-Cartan theory is the ability to include quantum corrections within a geometric theory of gravity by relating the affine connection with the matter intrinsic spin \cite{Hehl, Kibble, Sciama}. A review of the Einstein-Cartan theory can be found in \cite{Blagojevic}. The framework of the theory has led to many important results: In \cite{Luz_Carloni_2019} it was shown that torsion modifies the Buchdahl limit for the maximum compactness of static compact objects; Studies on gravitational collapse concluded that torsion can act as a repulsive force, possibly preventing the formation of singularities \cite{Luz_Mena_2019, Luz-Mena-Ziaie, Stewart, Trautman, Venn}; In cosmology, the presence of torsion and non-metricity, even at very low densities, may modify the dynamics, lead to the emission of characteristic gravitational waves or contribute to dark energy \cite{Barrow, Bolejko, Cubero, Venn2, Luz_Lemos_2023, Martins}.
In spite of the several studies on the Einstein-Cartan theory, a manifestly hyperbolic formulation of the field equations has not been found yet.
\\\\
The aim of this paper is to start the analysis of the initial value problem in the Einstein-Cartan theory. In particular we show a well-posedness result which in general terms can be written as:
\\\\
{\bf Theorem.} {\em Given sufficiently regular initial data on a Cauchy hypersurface, the initial value problem for the Einstein-Cartan system is locally well-posed, on a generalized harmonic gauge, and is future causal.}
\\\\
We do not impose particular matter fields and assume that the matter is given {\em ab initio} and is sufficiently regular. In this sense, we obtain {\em geometric} well-posedness. As we shall see, our result generalizes the well-known theorem of Choquet-Bruhat and York \cite{Choquet-Bruhat_York_1995, Choquet-Bruhat_York_1996} of General Relativity with an important difference: while in GR  the system can be cast in a Leray hyperbolic form \cite{Choquet-Bruhat_York_1996}, in the Einstein-Cartan theory, due to the torsion terms, the system is non-strictly hyperbolic of Leray-Ohya type. This means that for the corresponding Einstein-Cartan system, in general, the initial data will be in Gevrey classes, which seem to be less suitable for numerical simulations \cite{Reula}. See however more recent analytical developments about well-posedness in Gevrey classes \cite{Colombini} and numerical simulations in contexts outside GR, including fluid dynamics \cite{ Avalos-Gevrey, Chemin, Paicu}, diffusion and wave equations \cite{Chernov, Chu}, as well as results about the exponential convergence of numerical methods for Gevrey regularity \cite{Feischl}.
\\\\  
{\bf Plan of the paper}. In Section \ref{defs} we introduce useful definitions and fix our conventions. In Section \ref{embedd} we revise aspects of the geometry of embedded submanifolds with torsion and derive the generalized Ricci-Gauss-Codazzi equations. The results of those sections are purely geometric and do not use the field equations. In Section \ref{decomposition} we use a $n+1$ splitting of the Einstein-Cartan equations and derive the new constraints and evolution equation. Finally, in Section \ref{well-posed} we show that the Cauchy problem for the Einstein-Cartan equations is locally well-posed. We split the later section into three cases and summarize our main result in Theorem \ref{theorem1} of Subsection \ref{final-result}. The article also contains three appendices: In Appendix \ref{Appendix:np1_identities} we present  identities for the $n+1$ formalism for manifolds with torsion, in Appendix \ref{Appendix:conservation_laws} the $n+1$ decomposition of the conservation equations for a general stress-energy tensor and in Appendix \ref{sec:Wave-equation-for-K_(ab)} we derive a wave equation for the symmetric part of the second fundamental form.
\\\\
We use a geometrized units where $G=c=1$. Greek indices run from 0 to $N-1$ and Latin indices run from 1 to $N-1$.
\section{Definitions and conventions}
\label{defs}

Given the  variety of different conventions used in the literature,
we start by introducing the elementary definitions that will be used throughout this article.

Let $\left(\mathcal{M},g,S\right)$ be an $N$--dimensional Lorentzian
manifold, where $g$ represents the metric tensor and $S$ the torsion
tensor, defined below, endowed with a metric compatible affine connection,
$C_{\alpha\beta}^{\gamma}$, such that, the covariant derivative of
a vector field $u\in T \mathcal{M}$ 
is given by
\begin{equation}
\nabla_{\alpha}u^{\gamma}=\partial_{\alpha}u^{\gamma}+C_{\alpha\beta}^{\gamma}u^{\beta}\,\label{eq:Covariant_definition}
\end{equation}
The anti-symmetric part of the connection $C_{\alpha\beta}^{\gamma}$
defines a tensor field called torsion
\begin{equation}
S_{\alpha\beta}{}^{\gamma}:=C_{\left[\alpha\beta\right]}^{\gamma}:=\frac{1}{2}\left(C_{\alpha\beta}^{\gamma}-C_{\beta\alpha}^{\gamma}\right)\,.\label{eq:Torsion_tensor}
\end{equation}
Using this definition, it is possible to split the connection into
an appropriate combination of the torsion tensor plus the usual metric (Levi-Civita)
connection, $\Gamma_{\alpha\beta}^{\gamma}$, such that
\begin{equation}
C_{\alpha\beta}^{\gamma}=\Gamma_{\alpha\beta}^{\gamma}+S_{\alpha\beta}{}^{\gamma}+S^{\gamma}{}_{\alpha\beta}-S_{\beta}{}^{\gamma}{}_{\alpha}\,,\label{eq:Connection}
\end{equation}
with
\begin{equation}
\Gamma_{\alpha\beta}^{\gamma}=\frac{1}{2}g^{\gamma\sigma}\left(\partial_{\alpha}g_{\sigma\beta}+\partial_{\beta}g_{\alpha\sigma}-\partial_{\sigma}g_{\alpha\beta}\right)\,.\label{eq:Christoffel_symbols}
\end{equation}
The difference between the total metric compatible connection and
the Levi-Civita connection defines the contorsion tensor, 
\begin{equation}
{\cal K}_{\alpha\beta}{}^{\gamma}:=C_{\alpha\beta}^{\gamma}-\Gamma_{\alpha\beta}^{\gamma}=S_{\alpha\beta}{}^{\gamma}+S^{\gamma}{}_{\alpha\beta}-S_{\beta}{}^{\gamma}{}_{\alpha}\,.\label{eq:Contorsion}
\end{equation}
From the anti-symmetry of the torsion tensor in the first two indices,
it is straightforward to verify the following identities for the contorsion
tensor
${\cal K}_{\alpha\beta\gamma}  =-{\cal K}_{\alpha\gamma\beta} \label{eq:Contorsion_last2indices}$ and 
${\cal K}_{\left[\alpha\beta\right]}{}^{\gamma}  =S_{\alpha\beta}{}^{\gamma}\,.\label{eq:Contorsion_first2indices}$
Having in mind the general affine connection~(\ref{eq:Connection}),
the Lie derivative between two vector fields $u,v\in T\mathcal{M}$
can be expressed in terms of the torsion-full covariant derivative
as
\begin{equation}
\left(\mathcal{L}_{u}v\right)^{\gamma}=\left[u,v\right]^{\gamma}=u^{\alpha}\nabla_{\alpha}v^{\gamma}-v^{\alpha}\nabla_{\alpha}u^{\gamma}-2S_{\alpha\beta}\,^{\gamma}u^{\alpha}v^{\beta}\,.\label{eq:commutator_covariant_derivatives}
\end{equation}
This equation and the definition of the Riemann tensor associated
with $C_{\alpha\beta}^{\gamma}$, i.e.
\begin{equation}
R_{\alpha\beta\gamma}{}^{\rho}=\partial_{\beta}C_{\alpha\gamma}^{\rho}-\partial_{\alpha}C_{\beta\gamma}^{\rho}+C_{\beta\sigma}^{\rho}C_{\alpha\gamma}^{\sigma}-C_{\alpha\sigma}^{\rho}C_{\beta\gamma}^{\sigma}\,,\label{eq:Riemann_tensor_connnection}
\end{equation}
lead to a modified version of the relation between the Riemann curvature
tensor and the commutator of two covariant derivatives in the case
of non-vanishing torsion:
\begin{equation}
\begin{aligned}\sum_{i=1}^{m}R_{\alpha\beta\nu_{i}}{}^{\rho}\psi^{\mu_{1}\ldots\mu_{k}}{}_{\nu_{1}\ldots\rho\ldots\nu_{m}}-\hspace{0.8cm}\\
-\sum_{i=1}^{k}R_{\alpha\beta\rho}{}^{\mu_{i}}\psi^{\mu_{1}\ldots\rho\ldots\mu_{k}}{}_{\nu_{1}\ldots\nu_{m}} & =\left(\nabla_{\alpha}\nabla_{\beta}-\nabla_{\beta}\nabla_{\alpha}+2S_{\alpha\beta}{}^{\gamma}\nabla_{\gamma}\right)\psi^{\mu_{1}\ldots\mu_{k}}{}_{\nu_{1}\ldots\nu_{m}}\,.
\end{aligned}
\label{eq:Riemann_tensor_definition}
\end{equation}
where $\psi$ is an arbitrary $\left(k,m\right)$-tensor. We recall that in the case
of a Lorentzian manifold with non-vanishing torsion, the Riemann curvature
tensor does not possess the same symmetries as in the torsion-free
case, in particular it verifies $R_{\alpha\beta\gamma\delta}  =-R_{\beta\alpha\gamma\delta}$ and 
$R_{\alpha\beta\gamma\delta}  =-R_{\alpha\beta\delta\gamma}$ as well as the modified Bianchi identities
\begin{align}
R_{\left[\alpha\beta\gamma\right]}{}^{\delta} & =-2\nabla_{\left[\alpha\right.}S_{\left.\beta\gamma\right]}{}^{\delta}+4S_{\left[\alpha\beta\right|}{}^{\rho}S_{\left|\gamma\right]\rho}{}^{\delta}\,,\label{eq:first_Bianchi_identity}\\
\nabla_{\left[\alpha\right.}R_{\left.\beta\gamma\right]\delta}{}^{\rho} & =2S_{\left[\alpha\beta\right|}{}^{\sigma}R_{\left|\gamma\right]\sigma\delta\rho}\,.\label{eq:second_Bianchi_identity}
\end{align}
From Eq.~(\ref{eq:first_Bianchi_identity}) we find that the usual
symmetry of exchanging the first and second pair of indices of the
Riemann tensor is also modified in the presence of torsion, so that
\begin{equation}
2R_{\gamma\delta\alpha\beta}=2R_{\alpha\beta\gamma\delta}+3A_{\alpha\gamma\beta\delta}+3A_{\delta\alpha\beta\gamma}+3A_{\gamma\delta\alpha\beta}+3A_{\beta\delta\gamma\alpha}\,,\label{eq:Riemann_exchange_pair_indices}
\end{equation}
where $A_{\alpha\beta\gamma\delta}\equiv-2\nabla_{\left[\alpha\right.}S_{\left.\beta\gamma\right]\delta}+4S_{\left[\alpha\beta\right|}{}^{\rho}S_{\left|\gamma\right]\rho\delta}$.
Lastly, we recall that using Eq.~(\ref{eq:Riemann_tensor_connnection}) we have
the following expression for the Ricci tensor components in terms
of the connection and its derivatives
\begin{equation}
R_{a\beta}=R_{\alpha\gamma\beta}{}^{\gamma}=\partial_{\gamma}C_{\alpha\beta}^{\gamma}-\partial_{\alpha}C_{\gamma\beta}^{\gamma}+C_{\gamma\sigma}^{\gamma}C_{\alpha\beta}^{\sigma}-C_{\alpha\sigma}^{\gamma}C_{\gamma\beta}^{\sigma}\,.\label{eq:Ricci_tensor_connnection}
\end{equation}
%
\section{Geometry of embedded submanifolds with torsion\label{sec:Geometry_of_embedded_submanifolds}}
\label{embedd}
We will now consider the
existence of an embedding map $\Phi:\mathcal{V}\to \mathcal{M}$, where $\mathcal{V}$ is a differential manifold with  dimension\footnote{It is, nevertheless, possible to easily extend the upcoming results
to manifolds of dimension $N-m$, with $0<m<N$, by applying them
recursively or, equivalently, redefining some quantities.} $N-1$. We are interested
in studying how the geometrical structure on $\mathcal{M}$ induces
a geometrical structure on $\mathcal{V}$ and the relation between
the two.
\subsection{First and second fundamental forms}

\global\long\def\ps#1#2{\prescript{#1}{}{#2}}

Let $\xi^a$ be  a coordinate system on $\mathcal{V}$, where $a,b,...=1,2,...,N-1$ are indices on the hypersurface. 
Given a vector basis $\{\partial/\partial \xi^a\}$ of the tangent space $T\mathcal{V}$, 
the push-forward $d\Phi$ provides a correspondence between the vectors of this basis and sets of  linearly independent vectors $\left\{e_a \right\}$ tangent to
$\mathcal{V}$, given in  coordinates by $e^{\alpha}_a = \partial_{\xi^a} \Phi^{\alpha}$.
Here, we follow the convention of the previous section that greek letters $\alpha,\beta,...=0,1,2,...,N-1$ represent 
spacetime indices. We also consider (up to orientation) a vector $n^{\alpha}$ normal to the hypersurface $\mathcal{V}$ 
satisfying 
\begin{equation}
\begin{aligned}n_{\alpha}e_{a}^{\alpha} & =0\,,\\
n_{\alpha}n^{\alpha} & =\varepsilon\,,\\
n_{\alpha}\nabla_{\beta}n^{\alpha} & =0\,,
\end{aligned}
\end{equation}
where $\varepsilon=-1$ if the hypersurface $\mathcal{V}$
is spacelike and $\varepsilon=+1$ if the hypersurface $\mathcal{V}$
is timelike.
The first and second fundamental forms on $\mathcal{V}$ are given by $h=\Phi^{\star}(g)$ and $K=\Phi^{\star}(\nabla {\bf n})$, where 
$\Phi^{\star}$ denotes the pull-back corresponding to the map $\Phi$. In components, we may write
\begin{eqnarray}
&&h_{ab}:=e_{a}^{\alpha}e_{b}^{\beta}g_{\alpha\beta} \label{eq:first_fundamental_form}\\
&&K_{ab}:=e_{a}^{\alpha}e_{b}^{\beta}\nabla_{\alpha}n_{\beta}\,.\label{eq:extrinsic_curvature_definition}
\end{eqnarray}
Notice that the adopted
definition for $K_{ab}$ differs from the one in \cite{Choquet-Bruhat_York_1995,Choquet-Bruhat_Book_2009,Anderson_Choquet-Bruhat_York_1999}
by a minus sign.

Consider now a congruence of non-null curves
in $\mathcal{M}$ with
tangent vector field $n$, orthogonal at each point to $\mathcal{V}$.
Assuming, without loss of
generality, that the curves of the congruence are affinely parameterized,
we can introduce the projector, in $\mathcal{M}$, onto the orthogonal
space to the congruence as
\begin{equation}
h_{\alpha\beta}=g_{\alpha\beta}-\varepsilon n_{\alpha}n_{\beta}\,.\label{eq:projector_definition}
\end{equation}
The necessary and sufficient conditions for a manifold
to admit integral submanifolds orthogonal to a congruence are provided
by the Frobenius theorem.
Consider then that such conditions are verified and the congruence
associated with $n$ is such that the orthogonal space is a hypersurface $\mathcal{V}$.

The projector $h_{\alpha\beta}$ can be related with the induced metric
on $\mathcal{V}$ by $h_{ab}=e_{a}^{\alpha}e_{b}^{\beta}h_{\alpha\beta}\,\label{eq:projector_induced_metric}$
and  we can readily find the fully orthogonal
components of tensorial quantities on $\mathcal{M}$, that is, given
an arbitrary tensor $\psi^{\alpha...\beta}\in T\mathcal{M}\otimes...\otimes T\mathcal{M}$,
the tensor
\begin{equation}
\left(\psi_{\perp}\right)^{\alpha...\beta}=h_{\mu}{}^{\alpha}...h_{\nu}{}^{\beta}\psi^{\mu...\nu}\,,
\end{equation}
is completely orthogonal to the congruence. So, using the Jacobian
matrix, we can find the induced tensor on $\mathcal{V}$ associated
with $\left(\psi_{\perp}\right)^{\alpha...\beta}$,
\begin{equation}
\psi^{a...b}=e_{\alpha}^{a}...e_{\beta}^{b}\left(\psi_{\perp}\right)^{\alpha...\beta}\,.\label{eq:induced_tensor_definition}
\end{equation}
In addition to the induced metric, the map $\Phi$ also induces a connection
on $\mathcal{V}$. Given a tensor field $\psi_{\alpha...\beta}\in T^{*}\mathcal{M}\otimes...\otimes T^{*}\mathcal{M}$,
completely orthogonal to $n$, we define the intrinsic covariant derivative
of the tensor field $\psi_{a...b}:=e_{a}^{\alpha}...e_{b}^{\beta}\psi_{\alpha...\beta}$
on $\mathcal{V}$ as
\begin{equation}
D_{a}\psi_{b...c}:=e_{a}^{\alpha}e_{b}^{\beta}...e_{c}^{\gamma}\nabla_{\alpha}\psi_{\beta...\gamma}\,.\label{eq:induced_covariant_derivative1}
\end{equation}
Using this definition and Eq.~(\ref{eq:first_fundamental_form})
it is straightforward to show that the compatibility with the metric
$g$ of $\nabla$ induces $D$ to be compatible with the metric $h$, i.e. $D_{a}h_{bc}=e_{a}^{\alpha}e_{b}^{\beta}e_{c}^{\gamma}\nabla_{\alpha}g_{\beta\gamma}=0\,.$
 From Eq.~(\ref{eq:induced_tensor_definition}) we can write
\begin{equation}
D_{a}\psi_{b}=\partial_{a}\psi_{b}-C_{ab}^{c}\psi_{c}\,,\label{eq:induced_covariant_derivative2}
\end{equation}
where the induced connection is given by
\begin{equation}
C_{ab}^{c}=e_{a}^{\alpha}\left(\nabla_{\alpha}e_{b}^{\beta}\right)e_{\beta}^{c}\,.\label{eq:connection_induced_general_definition}
\end{equation}
Now, let $C_{ab}^{c}=\Gamma_{ab}^{c}+\mathcal{K}_{ab}{}^{c}$. Imposing that $D_{a}\psi_{b}$ is a tensor on $\mathcal{V}$ and
requiring that $\Gamma_{\left(ab\right)}^{c}=\Gamma_{ab}^{c}$, where the parenthesis in the indices denotes symmetrization, and
using Eq.~(\ref{eq:connection_induced_general_definition}) yields
\begin{align}
\Gamma_{ab}^{c} & =\frac{1}{2}h^{ci}\left(\partial_{a}h_{ib}+\partial_{b}h_{ai}-\partial_{i}h_{ab}\right)\,,\\
{\cal K}_{ab}{}^{c} & =S_{ab}{}^{c}+S^{c}{}_{ab}-S_{b}{}^{c}{}_{a}\,,
\end{align}
with
\begin{equation}
S_{ab}{}^{c}:=e_{a}^{\alpha}e_{b}^{\beta}e_{\gamma}^{c}S_{\alpha\beta}{}^{\gamma}=C_{\left[ab\right]}^{c}\,,\label{eq:induced_torsion_tensor}
\end{equation}
defining the induced torsion tensor on $\mathcal{V}$.

Computing the covariant derivative along a curve orthogonal to $n$
of a vector field $\psi\in T\mathcal{M}$, such that $\psi^{\alpha}n_{\alpha}=0$,
we find
\begin{equation}
e_{a}^{\beta}\nabla_{\beta}\psi^{\alpha}=e_{b}^{\alpha}D_{a}\psi^{b}-\varepsilon K_{ac}\psi^{c}n^{\alpha}\,,\label{eq:extrinsic_curvature_covariant_derivative_relation}
\end{equation}
showing how the second fundamental form measures the normal
component of the covariant derivative of a vector field, taken along
a direction orthogonal to $n$. Using Eq.~(\ref{eq:extrinsic_curvature_covariant_derivative_relation})
in the case when $\psi\equiv e_{b}$ leads to the Gauss-Weingarten equation
\begin{equation}
e_{a}^{\beta}\nabla_{\beta}e_{b}^{\alpha}=C_{ab}^{c}e_{c}^{\alpha}-\varepsilon K_{ab}n^{\alpha}\,.\label{eq:Gauss-Weingarten_equation}
\end{equation}
Although the previous results have the same form as the expressions
found for manifolds with symmetric, metric compatible, affine connections
(see e.g. \cite{Anderson_Choquet-Bruhat_York_1999}), the
presence of torsion fundamentally changes the properties of various
quantities. In particular, in the presence of torsion we have
\begin{equation}
K_{\left[ab\right]}=-e_{a}^{\alpha}e_{b}^{\beta}S_{\alpha\beta}{}^{\gamma}n_{\gamma}\,,\label{eq:antisymmetric_2nd_Fundamental_form_torsion_relation}
\end{equation}
that is, the second fundamental form is no longer required to be a
symmetric tensor.
%
\subsection{Gauss-Codazzi-Ricci embedding equations}
We are now in a position to study the
relation between the Riemann curvature tensor of $\mathcal{M}$ and
that of the submanifold $\mathcal{V}$, given by\footnote{More generally, the intrinsic Riemann tensor verifies the Ricci identity~(\ref{eq:Riemann_tensor_definition}),
considering the induced connection. }
\begin{equation}
\ps{\left(N-1\right)}{R_{abc}{}^{d}}\psi_{d}=D_{a}D_{b}\psi_{c}-D_{b}D_{a}\psi_{c}+2S_{ab}{}^{d}D_{d}\psi_{c}\,,
\end{equation}
where $\psi\in T^{*}\mathcal{V}$. From Eqs.~(\ref{eq:Riemann_tensor_definition}),
(\ref{eq:projector_definition}) and (\ref{eq:Gauss-Weingarten_equation}),
we find
\begin{align}
e_{a}^{\alpha}e_{b}^{\beta}e_{c}^{\gamma}e_{d}^{\delta}R_{\alpha\beta\gamma\delta} & =\text{\ensuremath{\ps{\left(N-1\right)}{R_{abcd}}}}-\varepsilon\left(K_{ac}K_{bd}-K_{bc}K_{ad}\right)\,,\label{eq:Gauss_equation}\\
\nonumber \\e_{a}^{\alpha}e_{b}^{\beta}e_{c}^{\gamma}n^{\delta}R_{\alpha\beta\gamma\delta} & =D_{a}K_{bc}-D_{b}K_{ac}+2S_{ab}{}^{d}K_{dc}\,.\label{eq:Codazzi_equation}
\end{align}
which are the Gauss and the Codazzi
equations, respectively, for manifolds with a non-vanishing torsion tensor field.

Note that the Gauss equation~(\ref{eq:Gauss_equation}) for Lorentzian
manifolds with torsion retains the same expression as in the torsionless
case, so that the difference between the pull-back of the fully projected
Riemann tensor on $\mathcal{M}$ and the intrinsic Riemann tensor
on $\mathcal{V}$ is completely described by the second fundamental
form. 

On the other hand, the presence of a general torsion field explicitly
modifies the Codazzi equation~(\ref{eq:Codazzi_equation}), contributing
to the tangent components of the Riemann tensor of $\mathcal{M}$.
These results explicitly show that the intrinsic geometry of the submanifold
$\mathcal{V}$ is described
not only by the first and second fundamental forms  on $\mathcal{V}$ but also by the
induced torsion tensor. In that sense, the induced
torsion could be regarded as a third fundamental form of $\mathcal{V}$.

Anticipating what follows and for the sake of clarity, it is useful
to introduce some notation. Given a tensor field $\psi$ in a vector bundle associated with the
tangent and cotangent bundles of $\mathcal{M}$, orthogonal to the
tangent vector field $n$, such that $\psi_{\alpha...\beta}{}^{\gamma\ldots\sigma}=e_{\alpha}^{a}\ldots e_{\beta}^{b}e_{c}^{\gamma}\ldots e_{d}^{\sigma}\psi_{a\ldots b}{}^{c\ldots d}$,
we define
\begin{equation}
\dot{\psi}_{a\ldots b}{}^{c\ldots d}:=n^{\mu}\partial_{\mu}\psi_{a\ldots b}{}^{c\ldots d}\,.
\end{equation}
Moreover, we will use the following shorthand notation
\begin{equation}
\begin{aligned}\psi_{\alpha\ldots0\ldots\beta}{}^{\gamma\ldots0\ldots\sigma} & :=n^{\mu}n_{\nu}\psi_{\alpha\ldots\mu\ldots\beta}{}^{\gamma\ldots\nu\ldots\sigma}\,,\\
\psi_{\alpha\ldots i\ldots\beta}{}^{\gamma\ldots j\ldots\sigma} & :=e_{i}^{\mu}e_{\nu}^{j}\psi_{\alpha\ldots\mu\ldots\beta}{}^{\gamma\ldots\nu\ldots\sigma}\,,
\end{aligned}
\end{equation}
to indicate partial contractions with the tangent vectors $e$ and
$n$. It is also useful to introduce the following tensors, defined from
the non-trivial contractions of the covariant derivative of the vector
fields $n$ and $e_{a}$:
\begin{equation}
\begin{aligned}a^{b} & :=e_{\mu}^{b}n^{\nu}\nabla_{\nu}n^{\mu}\,,\\
\varphi^{a}{}_{b} & :=e_{\alpha}^{a}n^{\beta}\nabla_{\beta}e_{b}^{\alpha}\,.
\end{aligned}
\label{eq:definition_a_varphi}
\end{equation}
\begin{lem}
Assuming that in some open set the manifold $\mathcal{M}$ is foliated
by a family of hypersurfaces $\left\{ \mathcal{V}\right\} $, orthogonal
at each point to a congruence of affinely parameterized curves with
tangent vector field $n$, the  Riemann
tensor on $\mathcal{M}$ verifies
\begin{align}
e_{a}^{\alpha}n^{\beta}e_{c}^{\gamma}n^{\delta}R_{\alpha\beta\gamma\delta} =& D_{a}a_{c}-K_{ad}K^{d}{}_{c}-\dot{K}_{ac}-\varepsilon\,a_{a}a_{c}+K_{ad}\varphi^{d}{}_{c}+K_{dc}\varphi^{d}{}_{a}\nonumber \\
 & +2n^{\mu}S_{a\mu}{}^{d}K_{dc}+2\varepsilon\,n^{\mu}n^{\nu}S_{a\mu\nu}a_{c}\,,\label{eq:Ricci_equation}
\end{align}
and
\begin{equation}
\label{eq:Second_Codazzi_equation}
\begin{aligned}
\frac{1}{2}e_{a}^{\alpha}&e_{b}^{\beta}e_{c}^{\gamma}n^{\delta}\left(R_{\gamma\delta\alpha\beta}-R_{\alpha\beta\gamma\delta}\right)=D_{a}S_{0\left(bc\right)}-D_{b}S_{0\left(ac\right)}+D_{c}S_{0\left[ba\right]}-\frac{1}{2}\dot{S}_{abc}+\dot{S}_{c\left[ab\right]}-\frac{3}{2}D_{\left[a\right.}K_{\left.bc\right]}\\&-2S_{0a}{}^{\mu}S_{\mu\left(bc\right)}+2S_{0b}{}^{\mu}S_{\mu\left(ac\right)}-S_{0ia}S_{bc}{}^{i}-2S_{\mu\left[ab\right]}S_{c0}{}^{\mu}-2S_{0i\left(b\right.}S_{\left.c\right)a}{}^{i}-\frac{3}{2}S_{\left[ab\right|}{}^{i}K_{i\left|c\right]}\\&+S_{i\left[ab\right]}K_{c}{}^{i}-S_{i\left(bc\right)}K_{a}{}^{i}+S_{i\left(ac\right)}K_{b}{}^{i}+\varepsilon S_{a00}K_{\left(bc\right)}-\varepsilon S_{\left[b\right|00}K_{\left|c\right]a}-\varepsilon K_{a\left(b\right.}S_{\left.c\right)00}+U_{abc}\,,
\end{aligned}
\end{equation}
where
\begin{equation}
	\label{eq:U_field}
	\begin{aligned}
		U_{abc}&=S_{\mu\left(bc\right)}\left(e_{i}^{\mu}\varphi^{i}{}_{a}-\varepsilon n^{\mu}a_{a}\right)+\left(S_{\mu\left[ba\right]}+\frac{1}{2}S_{ab\mu}\right)\left(e_{i}^{\mu}\varphi^{i}{}_{c}-\varepsilon n^{\mu}a_{c}\right)\\&+S_{\left[a\right|c\mu}\left(e_{i}^{\mu}\varphi^{i}{}_{\left|b\right]}-\varepsilon n^{\mu}a_{\left|b\right]}\right)-S_{\mu\left(ac\right)}\left(e_{i}^{\mu}\varphi^{i}{}_{b}-\varepsilon n^{\mu}a_{b}\right)\,,
	\end{aligned}
\end{equation}
following from the first Bianchi identity modified by the presence of torsion.
\end{lem}
Note that in the absence of torsion the left-hand-side of Eq.~(\ref{eq:Second_Codazzi_equation}) is identically zero, recovering the classical Codazzi equation. We remark that the tensor field \eqref{eq:U_field}  contains only terms directly related with the choice of the frame.
\begin{remark}
We note that a similar generalized Gauss equation (\ref{eq:Gauss_equation}) was derived in \cite{Boersma_Dray_1995}. 
As far as we know, the generalized Codazzi  (\ref{eq:Codazzi_equation}) and Ricci equations \eqref{eq:Ricci_equation} as well as Eq.~(\ref{eq:Second_Codazzi_equation}) were
never presented in the literature. Nonetheless, we remark that a particular version of these equations, derived in a different context, was given in \cite{Lau_1998} for a specific type of torsion.
\end{remark}
\begin{remark}
The above equations are purely geometric and do not assume the field equations of Einstein-Cartan theory. In fact, up to this point, our results apply to any metric compatible affine theory with or without torsion.
\end{remark}
\section{The $n+1$ decomposition of the field equations}
\label{decomposition}

The Gauss-Codazzi-Ricci equations
allow us to separate the Einstein-Cartan equations in a set of constraint
and evolution equations for the first and second fundamental forms as well as the induced torsion
on non-null hypersurfaces that locally foliate the spacetime manifold. 
In order to do that we use
of the $n+1$ spacetime formalism.
%
\subsection{The $n+1$ formalism}

Consider the spacetime ${\mathcal M}= {\mathcal V}\times \mathbb{R}$, such that ${\mathcal V}_\lambda={\mathcal V} \times \{\lambda\}$ are  either spacelike or timelike hypersurfaces. Later on, we studying the Cauchy problem, we will assume that ${\mathcal V}_\lambda$ are spacelike.
As in the last section, consider a frame $\{e_a\}$ tangent to ${\mathcal V}_\lambda$ and let  $t$ be a tensor field orthogonal to $\mathcal{V}_\lambda$, hence to $\{e_a\}$,
such that
\begin{equation}
\mathcal{L}_{t}e_{a}=0\,,\label{eq:propagate_tangent_vectors}
\end{equation}
where $\mathcal{L}_{t}$ represents the Lie derivative along $t$. Then, $\left(t, e_a \right) $ forms a local Cauchy adapted frame on $\mathcal{M}$.
In general, 
$t$ does not have the same direction
as $n$. However, at each point, $t$ can be decomposed as the sum
of a component parallel to $n$ and components tangent to $\mathcal{V}_{\lambda}$ as
\begin{equation}
t^{\alpha}=\mathsf{N}n^{\alpha}+\beta^{i}e_{i}^{\alpha}\,.
\end{equation}
The function $\mathsf{N}$ is usually called the \emph{lapse function} and
the vector field $\beta$ is called the \emph{shift vector}.

\subsection{Constraint and evolution equations}

For a $N$-dimensional manifold, with
$N>2$, the field equations of Einstein-Cartan theory can be derived from the Einstein-Hilbert action as
\begin{align}
G_{\alpha\beta}:=R_{\alpha\beta}-\frac{1}{2}g_{\alpha\beta}R & =8\pi\mathcal{T}_{\alpha\beta}-g_{\alpha\beta}\Lambda\,,\label{eq:efe1}\\
S^{\alpha\beta\gamma}+2g^{\gamma[\alpha}S^{\beta]}{}_{\mu}{}^{\mu} & =-8\pi\Delta^{\alpha\beta\gamma}\,,
\label{eq:efe2}
\end{align}
where $\mathcal{T}_{\alpha\beta}$
represents the components of the canonical stress-energy tensor, $\Lambda$
the cosmological constant, $\Delta_{\alpha\beta\mu}$ the components
of the intrinsic hypermomentum. Moreover, from the second
Bianchi identity~(\ref{eq:second_Bianchi_identity}) and the field
equations~(\ref{eq:efe1}) we find the conservation laws for the
canonical stress-energy tensor:
\begin{equation}
\nabla_{\beta}\mathcal{T}_{\alpha}{}^{\beta}=2S_{\alpha\mu\nu}\mathcal{T}^{\nu\mu}-\frac{1}{4\pi}S_{\alpha\mu}{}^{\mu}\Lambda+\frac{1}{8\pi}\left(S_{\alpha\mu}{}^{\mu}R-S^{\mu\nu\sigma}R_{\alpha\sigma\mu\nu}\right)\,.\label{eq:conservation_law}
\end{equation}
Using the results in Sec.~\ref{sec:Geometry_of_embedded_submanifolds},
we can find a set of constraint and evolution equations for the
dynamical variables. Let
\begin{equation}
\mathcal{L}_{0}:=\mathcal{L}_{t}-\mathcal{L}_{\mathsf{\beta}}\,,
\end{equation}
and consider the shorthand notation
$S_{a}=S_{a00}$, $K:=K_{i}{}^{i}$ to represent the trace of the
second fundamental form, and $\ps{\left(N-1\right)}{R_{ab}}:=\ps{\left(N-1\right)}{R_{aib}{}^{i}}$
to indicate the intrinsic Ricci tensor of each hypersurface $\mathcal{V}_\lambda$.
Then, we find the constraint equations
\begin{align}
2G_{00} & := K^{2}-K^{ab}K_{ba}-\varepsilon\,\,\ps{\left(N-1\right)}R=16\pi\mathcal{T}_{00}-2\varepsilon\Lambda\,,\label{eq:constraint_G00}\\
G_{a0} & := D_{i}K_{a}{}^{i}-D_{a}K-2S_{aij}K^{ji}=8\pi\mathcal{T}_{a0}\,,\label{eq:constraint_Gi0}\\
S_{abc} & =-8\pi\Delta_{abc}-\frac{16\pi}{N-2}h_{c\left[a\right|}\left(\Delta_{\left|b\right]i}{}^{i}+\varepsilon\Delta_{\left|b\right]00}\right)\,,
\label{eq:constraint_S_abc}\\
S_{ab0} & =-8\pi\Delta_{ab0}\,,\\
S_{0ab} & =\frac{8\pi}{N-2}h_{ab}\Delta_{0i}{}^{i}-8\pi\Delta_{0ab}\,,\\
S_{a} & =\frac{8\pi}{N-2}\left[\varepsilon\Delta_{ai}{}^{i}-\left(N-3\right)\Delta_{a00}\right]\,,\label{eq:constraint_S_a00}\\
K_{\left[ab\right]} & =-S_{ab0}\,,
\label{eq:constraint_K_S_relation}
\end{align}
and the evolution equations\footnote{We remark
that the notation $\psi_{\ldots0\ldots}$ differs from that in, for
instance, Ref.~\cite{Choquet-Bruhat_Book_2009} by a $\mathsf{N}^{-1}$
factor.}
\begin{eqnarray}
\label{eq:propagation_Gij}
R_{\alpha\beta}e_{a}^{\alpha}e_{b}^{\beta}&:=&  \ps{\left(N-1\right)}{R_{ab}}-\varepsilon K_{ab}K+\varepsilon K_{ai}K^{i}{}_{b}+\varepsilon K_{ai}K_{b}{}^{i}+\varepsilon D_{a}a_{b}-a_{a}a_{b}\nonumber\\
 && +2S_{a}a_{b}+2\varepsilon S_{0bi}K_{a}{}^{i}-\frac{\varepsilon}{\mathsf{N}}\mathcal{L}_{0}K_{ab}=8\pi\left(\mathcal{T}_{ab}-\frac{h_{ab}}{2}\mathcal{T}\right)+\frac{2 h_{ab}}{N-2} \Lambda\,,\\
\label{eq:propagation_G0i_antisymmetric}
G_{\left[0b\right]}&:=&\frac{1}{\mathsf{N}}\mathcal{L}_{0}S_{bi}{}^{i}+\frac{1}{\mathsf{N}}D_{i}\left(\mathsf{N}S_{0b}{}^{i}\right)-\frac{1}{\mathsf{N}}D_{b}\left(\mathsf{N}S_{0i}{}^{i}\right)+\frac{1}{\mathsf{N}}S^{i}{}_{b0}D_{i}\mathsf{N}\nonumber\\
 && +S_{bij}K^{ji}+2\varepsilon S_{bi0}S^{i}+\varepsilon S^{i}K_{bi}-\varepsilon S_{b}K=8\pi\mathcal{T}_{\left[0b\right]}\,,
\end{eqnarray}
where $\mathcal{T}:=\mathcal{T}_{\mu}{}^{\mu}=\mathcal{T}_{i}{}^{i}+\varepsilon\mathcal{T}_{00}$
represents the trace of the canonical stress-energy tensor. Furthermore, from the definition of the second fundamental form,
we have the evolution equation for the induced metric
\begin{equation}
	\mathcal{L}_{0}h_{ab}=2\mathsf{N}K_{\left(ab\right)}+4\mathsf{N}S_{0\left(ab\right)}.\label{eq:evolution_h}
\end{equation}
We note that 
Equation~(\ref{eq:constraint_K_S_relation}) is a consequence of
the Frobenius Theorem.
Additionally, to formalize the Cauchy problem for the geometric quantities and the matter fields, the conservation equations must
be included explicitly and are given in Appendix \ref{Appendix:conservation_laws}.
A few comments about the above equations are in order:
	\begin{itemize}
	\item Equations~(\ref{eq:constraint_G00})--(\ref{eq:evolution_h}) do
	not depend on the $\mathcal{L}_{0}$-derivatives of the variables
	$\mathtt{\mathsf{N}}$, $\beta$, $S_{a}$, $S_{0ab}$ and the traceless part of $S_{abc}$. Moreover,
	we see that the evolution equations~(\ref{eq:propagation_Gij}) and
	(\ref{eq:propagation_G0i_antisymmetric}), only depend on the $\mathcal{L}_{0}$
	derivative of the component $S_{ai}{}^{i}$ of the induced torsion.
		\begin{itemize}
			\item The lapse function $\mathtt{\mathsf{N}}$ and the shift
			vector $\beta$ are regarded
			as gauge variables.  
			Nonetheless, the lapse 
			is not completely free as from (\ref{eq:propagate_tangent_vectors})
			and the first Bianchi identity~(\ref{eq:first_Bianchi_identity})
			we find the  identities
			\begin{align}
				a_{b} & =2S_{b}-\frac{\varepsilon}{\mathsf{N}}D_{b}\mathsf{N}\,,
				\label{eq:constraint_a_N}\\
				\partial_{\left[a\right.}\partial_{\left.b\right]}\ln\left|\mathsf{N}\right| & =0\,.
				\label{eq:constraint_N_closed_form}
			\end{align}
			\item The previous equations show
			that the relation between the evolution of $S_{a}$, $S_{0ab}$, the traceless part of $S_{abc}$  and the evolution of the stress-energy tensor depends on the physical model. This leads to an interesting
			insight into the structure of the field equations and
			the relation between torsion and the matter source.

			\item Provided
			$\mathcal{T}_{[a0]}$, 
			Eq.~(\ref{eq:propagation_G0i_antisymmetric})
			are evolution equations for $S_{ai}{}^{i}$. From the conservation
			laws~(\ref{eq:propagation_L0_T00}) and (\ref{eq:propagation_L0_Ti0}) 
			we see that $\mathcal{L}_{0}\mathcal{T}_{00}$
			and $\mathcal{L}_{0}\mathcal{T}_{a0}$ are related to
			$\mathcal{L}_{0}S_{abc}$ and $\mathcal{L}_{0}K_{ab}$. However, these
			equations do not fully determine $\mathcal{L}_{0}S_{abc}$ and information from a particular physical model must be provided through the hypermomentum tensor. 
		\end{itemize}
	\end{itemize}

\begin{itemize}

\item The system (\ref{eq:constraint_G00})--(\ref{eq:evolution_h}) of second order partial differential equations for $h_{ab}$, $K_{ab}$,
and the torsion components is not a hyperbolic, in general. Moreover, it is not clear that the solutions depend solely on the
initial data in its causal past. In this regard, we will 
look for a hyperbolic formulation, locally equivalent to the Einstein-Cartan equations.

\end{itemize}

In conclusion, from a purely mathematical point of view, the components $S_{ab0}$ and $S_{ai}{}^{i}$
of the torsion tensor have a privileged standing, such that the field
equations only characterize the evolution of those components. Therefore, in general, to study the field equations
of the theory of Einstein-Cartan in the $n+1$ formalism, we can choose
the geometric variables
$$\left\{ h_{ab},K_{\left(ab\right)},S_{ai}{}^{i},S_{ab0}\right\} ,$$
where we explicitly decompose the extrinsic curvature in its symmetric
and anti-symmetric parts to clearly show the appearance of the $S_{ab0}$
component of the torsion field, and the matter related variables 
$$\left\{ \mathcal{T}_{00},\mathcal{T}_{ab},\mathcal{T}_{a0},\mathcal{T}_{0a}\right\}.$$
 The remaining components of the torsion tensor, through the hypermomentum tensor, can either be explicitly related to the unknowns by considering a particular matter model, or postulated functions	of the spacetime. Since the former case requires the adoption of a particular model, we will assume the latter case and keep the discussion as general as possible.

\subsection{Propagation of the constraints\label{sec:Propagation_contraints}}

In this section we show that the constraint equations~(\ref{eq:constraint_G00})
and (\ref{eq:constraint_Gi0}) are properly propagated, such that
if those constraints are verified at some initial spacelike slice, then they will be verified in subsequent slices, as in General Relativity.
\begin{prop}
\label{theorem:propagation_constraints_general}Let the evolution equations~(\ref{eq:propagation_Gij}) and (\ref{eq:propagation_G0i_antisymmetric})
as well as the conservation laws~(\ref{eq:propagation_L0_T00}) and (\ref{eq:propagation_L0_Ti0})
hold. If $\varepsilon=-1$ and the constraint equations~(\ref{eq:constraint_G00}) and
(\ref{eq:constraint_Gi0}) are verified at an initial hypersurface
$\mathcal{V}_{\lambda}$, then they are verified at $\left\{ \mathcal{V}_{\lambda+\delta\lambda}\right\} $,
for sufficiently small $\delta\lambda$.
\end{prop}
\begin{proof}
Consider the quantities
\begin{align}
\Sigma_{00} & :=G_{00}-8\pi\mathcal{T}_{00}+\varepsilon\Lambda\,,\label{eq_Econstraints:definition_Sigma_00}\\
\Sigma_{ab} & :=G_{ab}-8\pi\mathcal{T}_{ab}+h_{ab}\Lambda\,,\label{eq_Econstraints:definition_Sigma_ij}\\
\Sigma_{a0} & :=G_{a0}-8\pi\mathcal{T}_{a0}\,,\label{eq_Econstraints:definition_Sigma_i0}\\
\Sigma_{0a} & :=G_{0a}-8\pi\mathcal{T}_{0a}\,,\label{eq_Econstraints:definition_Sigma_0i}\\
F_{ab} & :=e_{a}^{\alpha}e_{b}^{\alpha}R_{\alpha\beta}-8\pi\left(\mathcal{T}_{ab}-\frac{h_{ab}}{2}\mathcal{T}\right)-\frac{2h_{ab}}{N-2}\Lambda.\label{eq_Econstraints:definition_F}
\end{align}
From the field equations~(\ref{eq:efe1}), we can write the constraint
equations~(\ref{eq:constraint_G00}) and (\ref{eq:constraint_Gi0})
as
\begin{align}
\Sigma_{00} & =\frac{1}{2}\left(K^{2}-K^{ab}K_{ba}-\varepsilon\,\,\ps{\left(N-1\right)}R\right)-8\pi\mathcal{T}_{00}+\varepsilon\Lambda\,,\label{eq_Econstraints:constraint_Sigma_00}\\
\Sigma_{a0} & =D_{i}K_{a}{}^{i}-D_{a}K-2S_{aij}K^{ji}-8\pi\mathcal{T}_{a0}\,.\label{eq_Econstraints:constraint_Sigma_a0}
\end{align}
Therefore, the constraint equations are verified if $\Sigma_{00}$
and $\Sigma_{a0}$ are identically zero. To evaluate the propagation
of these quantities it is useful to rewrite Eqs.~(\ref{eq:propagation_L0_T00})
and (\ref{eq:propagation_L0_Ti0}) in the more compact form
\begin{equation}
\begin{aligned}\frac{\varepsilon}{\mathsf{N}}\mathcal{L}_{0}\mathcal{T}_{00}= & -D_{i}\mathcal{T}_{0}{}^{i}+\mathcal{T}_{ij}K^{ji}-\varepsilon\mathcal{T}_{00}K+2\varepsilon\mathcal{T}_{\left(0i\right)}a^{i}-\frac{1}{8\pi}S^{ijk}R_{0kij}\\
 & -\frac{\varepsilon}{4\pi}S^{0ji}R_{0i0j}+2S_{0ij}\mathcal{T}^{ji}-2\varepsilon S_{i}\mathcal{T}^{0i}-S_{0i}{}^{i}\left(\frac{\Lambda}{4\pi}-\frac{R}{8\pi}\right)\,,
\end{aligned}
\label{eq_Econstraints:propagation_L0_T00}
\end{equation}
and
\begin{equation}
\begin{aligned}\frac{\varepsilon}{\mathsf{N}}\mathcal{L}_{0}\mathcal{T}_{a0}= & -D_{i}\mathcal{T}_{a}{}^{i}-\varepsilon\left(\mathcal{T}_{0i}K^{i}{}_{a}+\mathcal{T}_{a0}K\right)+\varepsilon\mathcal{T}_{ai}a^{i}+\left(S_{ai}{}^{i}+\varepsilon S_{a}\right)\left(\frac{R}{8\pi}-\frac{\Lambda}{4\pi}\right)\\
 & -\frac{1}{8\pi}S^{ijk}R_{akij}-\frac{\varepsilon}{8\pi}S^{ij0}R_{a0ij}-\frac{\varepsilon}{4\pi}S^{0ji}R_{ai0j}-\frac{1}{4\pi}S^{i}R_{a0i0}\\
 & +\frac{\varepsilon}{\mathsf{N}}\mathcal{T}_{00}D_{a}\mathsf{N}+\varepsilon\mathcal{T}_{i0}K_{a}{}^{i}+2S_{aij}\mathcal{T}^{ji}+2\varepsilon S_{ai0}\mathcal{T}^{0i}\,.
\end{aligned}
\label{eq_Econstraints:propagation_L0_Ti0}
\end{equation}
Moreover, the contracted second Bianchi identity yields
\begin{equation}
\begin{aligned}\frac{\varepsilon}{\mathsf{N}}\mathcal{L}_{0}G_{00}= & 2S_{0ij}R^{ji}-2\varepsilon S_{i}R_{0}{}^{i}-S^{ijk}R_{0kij}-D_{i}R_{0}{}^{i}-\varepsilon R_{00}K\\
 & -2\varepsilon S^{0ia}R_{0a0i}+\varepsilon\left(R_{0i}-R_{i0}\right)a^{i}+R_{ij}K^{ji}\,,
\end{aligned}
\label{eq_Econstraints:propagation_L0_G00}
\end{equation}
and
\begin{equation}
\begin{aligned}\frac{\varepsilon}{\mathsf{N}}\mathcal{L}_{0}R_{a0}= & -D_{i}R_{a}{}^{i}-\varepsilon R_{a0}K+\varepsilon R_{ai}a^{i}-R_{00}a_{a}-\varepsilon K_{a}{}^{i}\left(R_{0i}-R_{i0}\right)\\
 & -R_{abcd}S^{cdb}-2\varepsilon S^{0ji}R_{ai0j}-\varepsilon S^{ij0}R_{a0ij}-2S^{i}R_{a0i0}\\
 & +\frac{1}{2}D_{a}R+2S_{aij}R^{ji}+2S_{a}R_{00}\,.
\end{aligned}
\label{eq_Econstraints:propagation_L0_Gi0}
\end{equation}
Then, from Eqs.~(\ref{eq_Econstraints:propagation_L0_T00})--(\ref{eq_Econstraints:propagation_L0_Gi0})
we find the evolution equations for $\Sigma_{00}$ and $\Sigma_{a0}$:
\begin{equation}
\begin{aligned}
\label{eq_Econstraints:propagation_L0_Sigma_intermediate_1}
\mathcal{L}_{0}\Sigma_{00}= & 2\varepsilon\mathsf{N}S_{0ij}\Sigma^{ji}+2\mathsf{N}S_{i}\Sigma^{i}{}_{0}+\varepsilon\mathsf{N}\Sigma_{ij}K^{ji}-\varepsilon\mathsf{N}D_{i}\Sigma_{0}{}^{i}-\varepsilon\left(\Sigma_{0}{}^{i}+\Sigma^{i}{}_{0}\right)D_{i}\mathsf{N}-\mathsf{N}\Sigma_{00}K\,,
\\\mathcal{L}_{0}\Sigma_{a0}= & \Sigma_{00}D_{a}\mathsf{N}-\mathsf{N}(\varepsilon D_{i}\Sigma_{a}{}^{i}+\Sigma_{a0}K-2\varepsilon S_{aij}\Sigma^{ji})+\Sigma_{a}{}^{i}\left(2\mathsf{N}S_{i}-\varepsilon D_{i}\mathsf{N}\right)-\mathsf{N}K_{a}{}^{i}\left(\Sigma_{0i}-\Sigma_{i0}\right).
\end{aligned}
\end{equation}
To evaluate in what conditions is $\left(\Sigma_{00},\Sigma_{a0}\right)=\left(0,0\right)$
a solution of the system of equations~(\ref{eq_Econstraints:propagation_L0_Sigma_intermediate_1}),
it is useful to rewrite the equations in terms of $F_{ij}$, Eq.~(\ref{eq_Econstraints:definition_F}).
From the Eqs.~(\ref{eq_Econstraints:definition_Sigma_00}), (\ref{eq_Econstraints:definition_Sigma_ij})
and (\ref{eq_Econstraints:definition_F}) we find that the quantities
$F_{ij}$ and $\Sigma_{ij}$ are related by
\begin{equation}
\Sigma_{ab}=F_{ab}-h_{ab}\left(\mathrm{F}+\varepsilon\Sigma_{00}\right)\,,
\end{equation}
where $\mathrm{F}:=F_{i}{}^{i}$ represents the trace of the tensor
$F$. Then, using this relation in Eq.~(\ref{eq_Econstraints:propagation_L0_Sigma_intermediate_1})
yields
\begin{equation}
\begin{aligned}\frac{\varepsilon}{\mathsf{N}}\mathcal{L}_{0}\Sigma_{00}= & 2S_{0ji}\left[F^{ij}-h^{ij}\left(\mathrm{F}+\varepsilon\Sigma_{00}\right)\right]+K_{ji}\left[F^{ij}-h^{ij}\left(\mathrm{F}+\varepsilon\Sigma_{00}\right)\right]\\
 & +2\varepsilon S_{i}\Sigma^{i}{}_{0}-\varepsilon\Sigma_{00}K-\frac{1}{\mathsf{N}}\left(\Sigma_{0}{}^{i}+\Sigma^{i}{}_{0}\right)D_{i}\mathsf{N}-D_{i}\Sigma_{0}{}^{i}\,,
\\\frac{\varepsilon}{\mathsf{N}}\mathcal{L}_{0}\Sigma_{a0}= & -D_{i}F_{a}{}^{i}+D_{a}\mathrm{F}+\varepsilon D_{a}\Sigma_{00}-\varepsilon\Sigma_{a0}K-\varepsilon K_{a}{}^{i}\left(\Sigma_{0i}-\Sigma_{i0}\right)+\frac{\varepsilon}{\mathsf{N}}\Sigma_{00}D_{a}\mathsf{N}\\
 & +2S_{aji}\left[F^{ij}-h^{ij}\left(\mathrm{F}+\varepsilon\Sigma_{00}\right)\right]+\varepsilon\left[F_{a}{}^{i}-h_{a}{}^{i}\left(\mathrm{F}+\varepsilon\Sigma_{00}\right)\right]\left(2S_{i}-\frac{\varepsilon}{\mathsf{N}}D_{i}\mathsf{N}\right)\,.
\end{aligned}
\label{eq_Econstraints:propagation_L0_Sigma_intermediate_2}
\end{equation}
Note that these equations generalize those of the theory of general relativity
for the case of a general, metric compatible affine connection (cf.,
e.g., Ref.~\cite{Anderson_Choquet-Bruhat_York_1999}).

Now, the evolution equations~(\ref{eq:propagation_Gij}) and (\ref{eq:propagation_G0i_antisymmetric})
can be written as $F_{ab}=0$ and $\Sigma_{0a}-\Sigma_{a0}=0$, respectively.
Then, assuming $F_{ab}=0$ and $\Sigma_{0a}=\Sigma_{a0}$ hold, Eqs.~(\ref{eq_Econstraints:propagation_L0_Sigma_intermediate_2})
reduce to
\begin{equation}
\begin{aligned}\mathcal{L}_{0}\Sigma_{00}+\varepsilon\mathsf{N}D^{i}\Sigma_{i0}= & 2\mathsf{N}S^{i}\Sigma_{i0}-2\mathsf{N}S_{0i}{}^{i}\Sigma_{00}-2\mathsf{N}\Sigma_{00}K-2\varepsilon\Sigma_{i0}D^{i}\mathsf{N}\,,
\\\mathcal{L}_{0}\Sigma_{a0}-\mathsf{N}D_{a}\Sigma_{00}= & 2\Sigma_{00}D_{a}\mathsf{N}-\mathsf{N}\Sigma_{a0}K-2\mathsf{N}\left(S_{ai}{}^{i}+\varepsilon S_{a}\right)\Sigma_{00}\,.
\end{aligned}
\label{eq_Econstraints:propagation_L0_Sigma_final}
\end{equation}
or in matrix form,
\begin{equation}
\partial_{0}\mathsf{u}+\sum_{i=1}^{N-1}\mathsf{A}^{i}\partial_{i}\mathsf{u}=\mathcal{B}\left(\mathsf{u}\right)\,,\label{eq_Econstraints:propagation_L0_Sigma_final_matrix}
\end{equation}
where $\mathsf{u}=\left[\begin{array}{cccc}
\Sigma_{00} & \Sigma_{10} & \cdots & \Sigma_{N-1\,0}\end{array}\right]^{T}$, $\mathcal{B}$ is a column matrix whose entries can be read from
Eq.~(\ref{eq_Econstraints:propagation_L0_Sigma_final}) and 
\begin{equation}
\mathsf{A}^{i}=\left[\begin{array}{ccccc}
0 & \varepsilon\mathsf{N}h^{i1} & \varepsilon\mathsf{N}h^{i2} & \cdots & \varepsilon\mathsf{N}h^{i\,N-1}\\
-\mathsf{N}h_{1}{}^{i} & 0 & 0 & \cdots & 0\\
-\mathsf{N}h_{2}{}^{i} & 0 & 0 & \cdots & 0\\
\vdots & \vdots & \vdots & \ddots & \vdots\\
-\mathsf{N}h_{N-1}{}^{i} & 0 & 0 & \cdots & 0
\end{array}\right]\,.
\end{equation}
Introducing the $N\times N$ matrix
\begin{equation}
\mathsf{S}=\left[\begin{array}{cc}
1 & 0\\
0 & -\varepsilon h^{ab}
\end{array}\right]\,,
\end{equation}
where $h^{ab}$ represents the inverse of the induced metric on the
hypersurfaces $\left\{ \mathcal{V}_{\lambda}\right\} $, using the
fact that $h^{ai}\equiv h^{a}{}_{j}h^{ji}$ we find that the matrices
$\mathsf{S}\mathsf{A}^{i}$ are symmetric. Lastly, in the case when
$\varepsilon=-1$, that is, in the case when the normal vector field
$n$ is timelike, the hypersurfaces $\left\{ \mathcal{V}_{\lambda}\right\} $
are spacelike hence, $h_{ab}$ is a Riemannian metric, in particular
it is continuously positive definite and so is its inverse, $h^{ab}$.
Therefore, Eq.~(\ref{eq_Econstraints:propagation_L0_Sigma_final})
is a first-order, symmetrizable hyperbolic system hence, it has
a unique local solution. If the constraint equations are
verified at some initial hypersurface $\mathcal{V}_{\lambda}$, they
are verified at $\left\{ \mathcal{V}_{\lambda+\delta\lambda}\right\} $,
for sufficiently small $\delta\lambda$. The domain of dependence
is determined by the isotropic cone of the metric $g$.
\end{proof}
%

\section{Well-posedness of the Cauchy problem}
\label{well-posed}

Gathering the results of the previous sections, we are now in a position
to find explicit hyperbolic formulations for the field equations of
the Einstein-Cartan theory.
We will adopt the path taken in \cite{Choquet-Bruhat_York_1995, Choquet-Bruhat_York_1996} and derive second order evolution equations for the symmetric and antisymmetric parts of the second fundamental form. The derivation is rather lengthy, so we present most of the details in Appendix~\ref{sec:Wave-equation-for-K_(ab)}.

Due to the form of the field equations, if one of the torsion
components $S_{ai}{}^{i}$ or $S_{ab0}$, vanishes, the character
of the associated evolution equations changes, becoming a constraint equation.
Therefore, we have to distinguish between three cases: 
\begin{enumerate}
\item $S_{ai}{}^{i}$ and $S_{ab0}$ non-identicaly zero;
\item $S_{ai}{}^{i}$ identically zero; 
\item $S_{ab0}$ identically zero. 
\end{enumerate}
The case when both $S_{ai}{}^{i}$ and $S_{ab0}$ vanish is very special. In particular, it results in strong constraining conditions for the remaining
components of the torsion tensor which either have to be given {\em ab initio} as functions of
the spacetime or have to be directly related to the matter variables. So this case
will not be considered here.

Following the details of the derivation in Appendix \ref{sec:Wave-equation-for-K_(ab)}, it is convenient to use the gauge freedom and set
\begin{equation}
\label{gauge}
\mathcal{L}_{0}\mathsf{N}-\mathsf{N}^{2}K-2\mathsf{N}^{2}S_{0i}{}^{i}=0.
\end{equation}
Taking the trace of Eq.~(\ref{eq:evolution_h}) and using the
Jacobi's formula yields that the lapse function must be of the form
\begin{equation}
\mathsf{N}=\left|h\right|^{\frac{1}{2}}\alpha\left(t,x\right)\,,\text{ with }\text{\ensuremath{\mathcal{L}_{0}\alpha=0}}\,,
\label{eq_Wave:Algebraic_gauge}
\end{equation}
where $\alpha\left(t,x\right)$ is an otherwise arbitrary positive
scalar density, and $\left|h\right|$ represents the determinant of
the induced (Riemannian) metric $h$. Equation~(\ref{eq_Wave:Algebraic_gauge})
defines the so-called the \emph{algebraic gauge},
and the scalar density $\alpha$ is called the \emph{densitized lapse}.
\begin{remark}
	Instead of replacing the lapse function by the induced metric and the densitized lapse with (\ref{eq_Wave:Algebraic_gauge}), in practice we may include (\ref{gauge}) and treat the lapse as an additional unknown.
\end{remark}
%
\subsection{Case with $S_{ab0}$ and $S_{ai}{}^{i}$ not identically
zero}
\label{5.1}
In this case by introducing the tensor fields $P$ and $Q$ in Eqs.~\eqref{P_tensor} and \eqref{Q_tensor}, we can prove:
\begin{prop}
\label{theorem:hyperbolic_system_general_torsion}In the algebraic
gauge (\ref{eq_Wave:Algebraic_gauge}), let the densitised lapse $\alpha$, the shift vector $\beta$,
the torsion components $S_{a}$ and $S_{0ab}$, the traceless part
of the induced torsion $S_{abc}$, and the matter source canonical
stress-energy tensor $\mathcal{T}_{\alpha\beta}$ be in the Gevrey
class\footnote{The Gevrey classes are $C^\infty$ functions whose derivatives satisfy inequalities weaker than those satisfied by analytic functions \cite{Bruhat}. } of index 2. For an arbitrary cosmological constant $\Lambda$,
and for non-vanishing $S_{ab0}$ and $S_{ai}{}^{i}$, the system of
equations
\begin{equation}
	\begin{aligned}
		&\begin{aligned}-\mathsf{N}\square_{h} & K_{\left(ab\right)}+2D^{j}D_{\left(a\right|}\left(\mathsf{N}S_{0\left|b\right)j}\right)+2D^{j}D_{\left(a\right|}\left(\mathsf{N}S_{\left|b\right)j0}\right)-4D_{\left(a\right.}D_{\left.b\right)}\left(\mathsf{N}S_{0i}{}^{i}\right)+2D^{j}D_{\left(a\right|}\left(\mathsf{N}S_{0j\left|b\right)}\right)\\
			& +2D_{i}\mathcal{L}_{0}S^{i}{}_{\left(ab\right)}-2D^{i}D_{i}\left(\mathsf{N}S_{0\left(ab\right)}\right)+2D_{\left(a\right|}\mathcal{L}_{0}S_{\left|b\right)i}{}^{i}-2\mathcal{L}_{0}D_{\left(a\right.}S_{\left.b\right)}-\frac{2}{\mathsf{N}}S_{\left(a\right|}\mathcal{L}_{0}D_{\left|b\right)}\mathsf{N}\\
			& +P_{\left(ab\right)}+Q_{\left(ab\right)}=8\pi\mathcal{L}_{0}\left[\mathcal{T}_{\left(ab\right)}-\frac{h_{ab}}{N-2}\left(\mathcal{T}+\frac{\Lambda}{4\pi}\right)\right]-16\pi D_{\left(a\right|}\left(\mathsf{N}\mathcal{T}_{\left|b\right)0}\right)\,,
		\end{aligned}
		\\
		&\begin{aligned}-\frac{1}{\mathsf{N}}\mathcal{L}_{0}^{2} & S_{ab0}+4S_{ab}{}^{j}D_{j}\left(\mathsf{N}S_{0i}{}^{i}\right)-2\mathcal{L}_{0}D_{\left[a\right.}S_{\left.b\right]}+D_{i}\mathcal{L}_{0}S_{ab}{}^{i}\\
			& +2D_{\left[a\right|}\mathcal{L}_{0}S_{\left|b\right]i}{}^{i}+\frac{2}{\mathsf{N}}S_{\left[a\right|}\mathcal{L}_{0}D_{\left|b\right]}\mathsf{N}+P_{\left[ab\right]}+Q_{\left[ab\right]}=8\pi\mathcal{L}_{0}\mathcal{T}_{\left[ab\right]}\,,
		\end{aligned}
		\\
		&\begin{aligned}\frac{1}{\mathsf{N}}\mathcal{L}_{0} & S_{bi}{}^{i}+\frac{1}{\mathsf{N}}D_{i}\left(\mathsf{N}S_{0b}{}^{i}\right)-\frac{1}{\mathsf{N}}D_{b}\left(\mathsf{N}S_{0i}{}^{i}\right)+\frac{1}{\mathsf{N}}S^{i}{}_{b0}D_{i}\mathsf{N}\\
			& +S_{bij}K^{ji}-2S_{bi0}S^{i}-S^{i}K_{bi}+S_{b}K=8\pi\mathcal{T}_{\left[0b\right]}\,,
		\end{aligned}
		\\
		&\mathcal{L}_{0}h_{ab}=2\mathsf{N}K_{\left(ab\right)}+4\mathsf{N}S_{0\left(ab\right)}\,,
		\label{eq_hyperbolic_system:general_torsion}
	\end{aligned}
\end{equation}
form a non-strictly hyperbolic  system for the variables $h_{ab}$,
$K_{\left(ab\right)}$, $S_{ab0}$ and $S_{ai}{}^{i}$. If the Cauchy
data in a slice $\mathcal{V}$ is in the Gevrey class of
index 2, the Cauchy problem has a unique solution in a neighborhood
of $\mathcal{V}$, with domain of dependence determined
by the isotropic cone of the metric $g$.
\end{prop}
\begin{proof}
Using the field equations~(\ref{eq:efe1}) and setting $\varepsilon=-1$,
Eqs.~(\ref{eq:propagation_G0i_antisymmetric}), (\ref{eq:evolution_h}),
(\ref{eq_Wave:Wave_eq_densitized_symmetric}), (\ref{eq_Wave:Wave_eq_densitized_antisymmetric})
form the system (\ref{eq_hyperbolic_system:general_torsion}). We show that the system can be cast in a form for which local well-posedness is guaranteed.

Note that the tensor $P_{\left(ab\right)}$, Eq.~\eqref{P_tensor}, appearing in the system contains terms with the intrinsic
Riemann and Ricci tensors, namely $\ps{\left(N-1\right)}{R_{i\left(ab\right)}{}^{j}}K_{j}{}^{i}$
and $\ps{\left(N-1\right)}{R_{\left(a\right|}{}^{j}}K_{\left|b\right)j}$.
These tensors depend on derivatives up to second order of the induced
metric $h$, and derivatives up to first order of the induced torsion.
Moreover, Eq.~(\ref{eq_Wave:Wave_eq_densitized_symmetric}) explicitly
contains terms with the second order derivatives of $S_{ab0}$ and
$S_{ai}{}^{i}$. Nonetheless, the principal operator of the system~(\ref{eq_hyperbolic_system:general_torsion})
can be written as a triangular matrix with diagonal elements either
$\square_{h}$, $\mathcal{L}_{0}^{2}$ or $\mathcal{L}_{0}$, where
we have omitted multiplicative $\mathsf{N}$ factors. Then, choosing
the value 2 for the Leray-Volevic weights\footnote{Leray-Volevic weights are useful integer quantities to characterize the principal operator of quasilinear partial differential equations \cite{Leray_lecture_notes_1953,Bruhat}.} for all the unknowns, and
the weights
\begin{equation}
\begin{aligned}\begin{aligned}n[\mathcal{L}_{0}S_{bi}{}^{i}] & =n[\mathcal{L}_{0}h_{ab}]=1\,, & n[\square_{h}K_{\left(ab\right)}] & =n[\mathcal{L}_{0}^{2}S_{ab0}]=0\,,\end{aligned}
\end{aligned}
\end{equation}
for the equations, we conclude that the system is quasilinear. 

Now to complete the proof that the system is Cauchy regular we compute the characteristic determinant of the principal part. The
characteristic determinant for each pair $\left(ab\right)$ is $\mathsf{N}^{-1}\left(-\mathsf{N}^{-2}\xi_{0}^{2}+h^{ij}\xi_{i}\xi_{j}\right)\left(\xi_{0}\right)^{4}$
or $-\mathsf{N}\left(-\mathsf{N}^{-2}\xi_{0}^{2}+h^{ij}\xi_{i}\xi_{j}\right)\left(\xi_{0}\right)^{2}$,
when $a\neq b$ or $a=b$, respectively. In either case, given the
multiplicity of the characteristics, the system is non-strictly hyperbolic 
in the Leray-Ohya sense (see \cite{Choquet-Bruhat_Book_2009}
for a review). For initial data in the Gevrey class of index 2, the
general theory of Leray and Ohya guarantees the local existence and
uniqueness of a solution. The characteristics at a point are the isotropic
cone of the metric $g$ and the normal to the slice 
containing the point. Since the normal is interior to the cone, the
domain of dependence of the solution is the light cone.
\end{proof}
%
\subsection{Case with $S_{ab0}$ identically zero}
\label{5.2}

It is useful
to introduce the following decomposition. Let $\bar{S}_{abc}$
represent the traceless part of the induced torsion tensor, $S_{abc}$,
such that
\begin{equation}
\bar{S}_{abc}:=S_{abc}+\frac{2}{N-2}h_{c[a}S_{b]i}{}^{i}\,.
\end{equation}
In that case, we find the following result:
\begin{lem}
\label{lemma:Laplace_Saii}For a spacetime manifold $\left(\mathcal{M},g,S\right)$
of dimension $N>3$, if the components $S_{ab0}$ are identically
zero, $S_{ai}{}^{i}$ verify the following constraint equations
\begin{equation}
\begin{aligned}
\tilde{D}^{a}\tilde{D}_{a}S_{bi}{}^{i}&-D_{b}D^{a}S_{ai}{}^{i}
+\frac{N-2}{N-3}\left(D^{a}D_{i}\bar{S}_{ab}{}^{i}-D^{a}D_{[a}a_{b]}\right)+F_{b}\\= & \frac{N-2}{N-3}8\pi D^{a}\mathcal{T}_{\left[ab\right]}+8\pi S^{ji}{}_{i}\mathcal{T}_{bj}
  +S_{bi}{}^{i}\left(\frac{2}{N-2}\Lambda-4\pi\mathcal{T}\right)\,,
\end{aligned}
\label{eq_hyperbolic_system:constraint_trace_induced_torsion}
\end{equation}
where $\tilde{D}$ represents the Levi-Civita connection associated
with $h_{ab}$, and the 1-form
\begin{align}
F_{b}= & \left(K_{ab}K-2K^{i}{}_{a}K_{bi}-a_{a}a_{b}+2a_{a}S_{b}-2S_{0ai}K_{b}{}^{i}-D_{b}a_{a}+\frac{1}{\mathsf{N}}\mathcal{L}_{0}K_{ab}\right)S^{ai}{}_{i}\nonumber \\
 & +\left(\frac{N-2}{N-3}\right)D^{a}\left(2S_{ab}{}^{j}S_{ji}{}^{i}-2K_{ai}K^{i}{}_{b}+2S_{[a}a_{b]}+2S_{0\left[a\right|i}K_{\left|b\right]}{}^{i}\right)\\
 & -2K_{kb}{}^{j}S^{kc}{}_{c}S_{ji}{}^{i}-K_{cbk}K^{ckj}S_{ji}{}^{i}-2S^{jk}{}_{k}D_{j}S_{bi}{}^{i}-K_{kb}{}^{j}D^{k}S_{ji}{}^{i}\nonumber \\
 & -2S_{b}{}^{jk}D_{k}S_{ji}{}^{i}-D^{a}\left(K_{ab}{}^{j}S_{ji}{}^{i}\right)\,,\nonumber 
\end{align}
contains terms with only first order derivatives of the
unknowns.
\end{lem}
\begin{proof}
When all components $S_{ab0}$ vanish identically, we show that
if the spacetime is of dimension $N>3$, Eq.~(\ref{eq_Wave:Wave_eq_densitized_antisymmetric}) or, equivalently, the antisymmetric part of Eq.~(\ref{eq:propagation_Gij})
yields a constraint equation for the component $S_{ai}{}^{i}$ of
the induced torsion.

Let the components $S_{ab0}$ vanish identically. From Eq.~(\ref{Appendix_eq:Full_Ricci_metric_Ricci_relation})
we find
\begin{equation}
\ps{\left(N-1\right)}{R_{\left[ab\right]}}=D_{i}S_{ab}{}^{i}-D_{[a|}K_{i|b]}{}^{i}+2S_{i[a|}{}^{j}K_{j|b]}{}^{i}+K_{i[b|}{}^{j}K_{|a]j}{}^{i}-K_{ij}{}^{i}S_{ab}{}^{j}\,,
\end{equation}
where $K_{ab}{}^{c}$ represent the components of the induced contorsion
on an hypersurface $\mathcal{V}$. Then, using the field equations~(\ref{eq:propagation_Gij})
for $\varepsilon=-1$ we find
\begin{equation}
\begin{aligned}\frac{N-3}{N-2}\left(D_{a}S_{bi}{}^{i}-D_{b}S_{ai}{}^{i}\right)+D_{i}\bar{S}_{ab}{}^{i}+2S_{ab}{}^{j}S_{ji}{}^{i}\\
-2K_{ai}K^{i}{}_{b}-D_{[a}a_{b]}+2S_{[a}a_{b]}+2S_{0\left[a\right|i}K_{\left|b\right]}{}^{i} & =8\pi\mathcal{T}_{\left[ab\right]}\,.
\end{aligned}
\label{eq_hyperbolic_system:constraint_trace_induced_torsion_intermediate}
\end{equation}
Taking the $D^{a}$ derivative of this equation, and considering
Eqs.~(\ref{eq:Riemann_tensor_definition}) and (\ref{eq:propagation_Gij}),
and the fact that $K_{ab}=K_{\left(ab\right)}$, implies Eq.~(\ref{eq_hyperbolic_system:constraint_trace_induced_torsion}).
\end{proof}
\begin{remark} We remark that for a spacetime of dimension $N=3$, the traceless
part of the induced torsion, $\bar{S}_{abc}$, is trivial and Eq.~(\ref{eq_hyperbolic_system:constraint_trace_induced_torsion_intermediate})
becomes an algebraic constraint equation for $S_{ai}{}^{i}$.
\end{remark}
Lemma~\ref{lemma:Laplace_Saii} allows us to find a quasilinear wave
equation for $S_{ai}{}^{i}$, provided an integrability condition.
Taking the $\mathcal{L}_{0}$ derivative of Eq.~(\ref{eq:propagation_G0i_antisymmetric})
and using Eq.~(\ref{eq_hyperbolic_system:constraint_trace_induced_torsion})
we find
\begin{equation}
\begin{aligned}
\square_{h}S_{bi}{}^{i}&-D_{b}D^{a}S_{ai}{}^{i}-\frac{1}{\mathsf{N}}\mathcal{L}_{0}D_{i}S_{0b}{}^{i}+\frac{1}{\mathsf{N}}\mathcal{L}_{0}D_{b}S_{0i}{}^{i}-\frac{1}{\mathsf{N}^{2}}S_{0b}{}^{i}\mathcal{L}_{0}D_{i}\mathsf{N}+\frac{1}{\mathsf{N}^{2}}S_{0i}{}^{i}\mathcal{L}_{0}D_{b}\mathsf{N}\\
&+\frac{N-2}{N-3}\left(D^{a}D_{i}\bar{S}_{ab}{}^{i}-D^{a}D_{[a}a_{b]}\right)+{\mathcal  G}_{b} =\frac{N-2}{N-3}8\pi D^{a}\mathcal{T}_{\left[ab\right]}-\frac{8\pi}{\mathsf{N}}\mathcal{L}_{0}\mathcal{T}_{\left[0b\right]}\\
 & +8\pi S^{ji}{}_{i}\mathcal{T}_{bj}+S_{bi}{}^{i}\left(\frac{2}{N-2}\Lambda-4\pi\mathcal{T}\right)\,,
\end{aligned}
\end{equation}
where
\begin{align}
{\mathcal G}_{b} & =F_{b}-\frac{1}{\mathsf{N}^{3}}\left(\mathcal{L}_{0}\mathsf{N}\right)\left(S_{0i}{}^{i}D_{b}\mathsf{N}-S_{0b}{}^{i}D_{i}\mathsf{N}-\mathcal{L}_{0}S_{bi}{}^{i}\right)
 -\frac{1}{\mathsf{N}^{2}}\left(D_{i}\mathsf{N}\right)\left(\mathcal{L}_{0}S_{0b}{}^{i}\right)\\
 &+\frac{1}{\mathsf{N}^{2}}\left(D_{b}\mathsf{N}\right)\left(\mathcal{L}_{0}S_{0i}{}^{i}\right)
  -\frac{1}{\mathsf{N}}\mathcal{L}_{0}\left(S_{bij}K^{ji}-S^{i}K_{bi}+S_{b}K\right)\,,\nonumber 
\end{align}
contains terms with only first order derivatives of the
unknowns.

The term $D_{b}D^{a}S_{ai}{}^{i}$ spoils the hyperbolicity
of the equation. In Appendix~\ref{sec:Wave-equation-for-K_(ab)}, to
find a wave equation for the components $K_{\left(ab\right)}$ of
the second fundamental form, we made use of the gauge freedom to constraint
the lapse function. Therefore, the term $D_{b}D^{a}S_{ai}{}^{i}$
cannot be further gauged away without spoiling the hyperbolicity of
the equation for $K_{\left(ab\right)}$. Nonetheless, from a physical
point of view, it is reasonable to assume that the divergence $\tilde{D}_{a}S^{ai}{}_{i}$
is either a known function of the spacetime or related with the matter
 variables when considering a particular matter model.
So we assume that $\tilde{D}^{a}S_{ai}{}^{i}$
or, similarly, the quantity$D^{a}S_{ai}{}^{i}$ is a known differentiable
function of the spacetime. In that case we can prove the following:
\begin{prop}
\label{prop2}
Consider a spacetime manifold of dimension $N>3$. Let $D^{a}S_{ai}{}^{i}=Q$,
where $Q$ is some known function of the spacetime in the Gevrey class
of index 2. In the algebraic gauge (\ref{eq_Wave:Algebraic_gauge}), let the densitized lapse $\alpha$,
the shift vector $\beta$, the torsion components $S_{a}$, $S_{0ab}$
and $\bar{S}_{abc}$, and the matter source canonical stress-energy
tensor $\mathcal{T}_{\alpha\beta}$ in the Gevrey class of index 2.
For an arbitrary cosmological constant $\Lambda$, and for vanishing
$S_{ab0}$, the system of equations
\begin{equation}
	\begin{aligned}
		&\begin{aligned}-\mathsf{N}\square_{h} & K_{\left(ab\right)}+2D_{\left(a\right|}\mathcal{L}_{0}S_{\left|b\right)i}{}^{i}+2D^{j}D_{\left(a\right|}\left(\mathsf{N}S_{0\left|b\right)j}\right)+2D_{i}\mathcal{L}_{0}S^{i}{}_{\left(ab\right)}-2D^{i}D_{i}\left(\mathsf{N}S_{0\left(ab\right)}\right)\\
			& -2\mathcal{L}_{0}D_{\left(a\right.}S_{\left.b\right)}-4D_{\left(a\right.}D_{\left.b\right)}\left(\mathsf{N}S_{0i}{}^{i}\right)+2D^{j}D_{\left(a\right|}\left(\mathsf{N}S_{0j\left|b\right)}\right)-\frac{2}{\mathsf{N}}S_{\left(a\right|}\mathcal{L}_{0}D_{\left|b\right)}\mathsf{N}\\
			& +P_{\left(ab\right)}+Q_{\left(ab\right)}=8\pi\mathcal{L}_{0}\left[\mathcal{T}_{\left(ab\right)}-\frac{h_{ab}}{N-2}\left(\mathcal{T}+\frac{\Lambda}{4\pi}\right)\right]-16\pi D_{\left(a\right|}\left(\mathsf{N}\mathcal{T}_{\left|b\right)0}\right)
		\end{aligned}
		\\
		&\begin{aligned}\square_{h}S_{bi}{}^{i} & -D_{b}Q-\frac{1}{\mathsf{N}}\mathcal{L}_{0}D_{i}S_{0b}{}^{i}+\frac{1}{\mathsf{N}}\mathcal{L}_{0}D_{b}S_{0i}{}^{i}-\frac{1}{\mathsf{N}^{2}}S_{0b}{}^{i}\mathcal{L}_{0}D_{i}\mathsf{N}+\frac{1}{\mathsf{N}^{2}}S_{0i}{}^{i}\mathcal{L}_{0}D_{b}\mathsf{N}\\
			& +\frac{N-2}{N-3}\left(D^{a}D_{i}\bar{S}_{ab}{}^{i}-D^{a}D_{[a}a_{b]}\right)+G_{b}=\frac{N-2}{N-3}8\pi D^{a}\mathcal{T}_{\left[ab\right]}-\frac{8\pi}{\mathsf{N}}\mathcal{L}_{0}\mathcal{T}_{\left[0b\right]}\\
			& +8\pi S^{ji}{}_{i}\mathcal{T}_{bj}+S_{bi}{}^{i}\left(\frac{2}{N-2}\Lambda-4\pi\mathcal{T}\right)\,,
		\end{aligned}
		\\
		&\begin{aligned}\mathcal{L}_{0}h_{ab} & =2\mathsf{N}K_{\left(ab\right)}+4\mathsf{N}S_{0\left(ab\right)}\,,
		\end{aligned}
		\label{eq_hyperbolic_system:vanishing_Sab0_torsion_components}
	\end{aligned}
\end{equation}
form a non-strictly hyperbolic system for the variables $h_{ab}$, $K_{\left(ab\right)}$
and $S_{ai}{}^{i}$. If the Cauchy data in a slice $\mathcal{V}$
is in the Gevrey class of index 2, the Cauchy problem has a unique
solution in a neighborhood of $\mathcal{V}$, with domain
of dependence determined by the isotropic cone of the metric $g$.
\end{prop}
\begin{proof}
The principal operator of the system~(\ref{eq_hyperbolic_system:vanishing_Sab0_torsion_components})
can be written as a triangular matrix with diagonal elements either
$\square_{h}$ or $\mathcal{L}_{0}$, where we have omitted multiplicative
$\mathsf{N}$ factors. Then, choosing the value 2 for the Leray-Volevic
weights for all the unknowns, and the following weights for the equations
\begin{equation}
\begin{aligned}\begin{aligned}n[\square_{h}K_{\left(ab\right)}] & =n[\square_{h}S_{bi}{}^{i}]=0\,, & n[\mathcal{L}_{0}h_{ab}] & =1\,,\end{aligned}
\end{aligned}
\end{equation}
equations~(\ref{eq_hyperbolic_system:vanishing_Sab0_torsion_components})
form a quasilinear system. The characteristic determinant for each
pair $\left(ab\right)$ is $-\mathsf{N}\left(-\mathsf{N}^{-2}\xi_{0}^{2}+h^{ij}\xi_{i}\xi_{j}\right)^{2}\xi_{0}$.
The rest of the proof follows as in Proposition \ref{theorem:hyperbolic_system_general_torsion}. 
\end{proof}
%
\subsection{Case with $S_{ai}{}^{i}$ identically zero}
\label{5.3}
In this case, Eq.~(\ref{eq:propagation_G0i_antisymmetric})
becomes a constraint equation:
\begin{lem}
If $S_{ai}{}^{i}=0$, the field equation~(\ref{eq:propagation_G0i_antisymmetric})
reads
\begin{equation}
\begin{aligned}\frac{1}{\mathsf{N}}\left( D_{i}\left(\mathsf{N}S_{0b}{}^{i}\right)-D_{b}\left(\mathsf{N}S_{0i}{}^{i}\right)+S^{i}{}_{b0}D_{i}\mathsf{N}\right)
+\bar{S}_{bij}K^{ji}-S^{i}(2S_{bi0}+K_{bi})+S_{b}K & =8\pi\mathcal{T}_{\left[0b\right]}
\end{aligned}
\label{eq_hyperbolic_system:trace_Saii_zero_constraint}
\end{equation}
\end{lem}
This equation can be seen either as an algebraic constraint between
$D_{a}\mathsf{N}$, $S_{ab0}$ and $\mathcal{T}_{\left[0b\right]}$
or a differential constraint for the remaining torsion components,
depending, respectively, on whether those components are known postulated
functions of the spacetime or explicitly related to the matter fields. Without a particular matter
model, the latter case can not be further developed. In the former
case, that is, if the torsion components $S_{0ab}$, $S_{a}$ and
the traceless part of the induced torsion, $\bar{S}_{abc}$, are known
functions of the spacetime, Eq.~(\ref{eq_hyperbolic_system:trace_Saii_zero_constraint})
connects the choice of the lapse function, $\mathsf{N}$, with the
canonical stress-energy tensor and the $S_{ab0}$ components. Namely,
the choice of $\mathsf{N}$ is arbitrary and is related with the choice
of frame to describe dynamics of the matter fluid. Equation~(\ref{eq_hyperbolic_system:trace_Saii_zero_constraint})
then asserts that this choice must be related to the difference between
the energy-momentum fluxes $\mathcal{T}_{\alpha0}$ and $\mathcal{T}_{0\alpha}$,
perceived by the observer in that frame. Therefore, this equation
is not a direct constraint on $S_{ab0}$, but a relation between
the description of the matter fluid and the dynamics of the observer's
frame. Then, we are left with the following system for the unknowns:
\begin{prop}
\label{prop3}
In the algebraic gauge (\ref{eq_Wave:Algebraic_gauge}), let the densitized lapse $\alpha$, the shift
vector $\beta$, the torsion components $S_{a}$ and $S_{0ab}$, the
traceless part of the induced torsion $S_{abc}$, and the matter source
canonical stress-energy tensor $\mathcal{T}_{\alpha\beta}$ be in the
Gevrey class of index 2. For an arbitrary cosmological constant $\Lambda$,
and for vanishing $S_{ai}{}^{i}$, the system of equations given by
\begin{equation}
	\begin{aligned}
		&\begin{aligned}-\mathsf{N}\square_{h} & K_{\left(ab\right)}+2D^{j}D_{\left(a\right|}\left(\mathsf{N}S_{0\left|b\right)j}\right)+2D^{j}D_{\left(a\right|}\left(\mathsf{N}S_{\left|b\right)j0}\right)-4D_{\left(a\right.}D_{\left.b\right)}\left(\mathsf{N}S_{0i}{}^{i}\right)+2D_{i}\mathcal{L}_{0}S^{i}{}_{\left(ab\right)}\\
			& +2D^{j}D_{\left(a\right|}\left(\mathsf{N}S_{0j\left|b\right)}\right)-2D^{i}D_{i}\left(\mathsf{N}S_{0\left(ab\right)}\right)-2\mathcal{L}_{0}D_{\left(a\right.}S_{\left.b\right)}-\frac{2}{\mathsf{N}}S_{\left(a\right|}\mathcal{L}_{0}D_{\left|b\right)}\mathsf{N}\\
			& +P_{\left(ab\right)}+Q_{\left(ab\right)}=8\pi\mathcal{L}_{0}\left[\mathcal{T}_{\left(ab\right)}-\frac{h_{ab}}{N-2}\left(\mathcal{T}+\frac{\Lambda}{4\pi}\right)\right]-16\pi D_{\left(a\right|}\left(\mathsf{N}\mathcal{T}_{\left|b\right)0}\right)\,,
		\end{aligned}
		\\
		&\begin{aligned}-\frac{1}{\mathsf{N}}\mathcal{L}_{0}^{2}S_{ab0}+4S_{ab}{}^{j}D_{j}\left(\mathsf{N}S_{0i}{}^{i}\right) & -2\mathcal{L}_{0}D_{\left[a\right.}S_{\left.b\right]}+D_{i}\mathcal{L}_{0}S_{ab}{}^{i}\\
			& +\frac{2}{\mathsf{N}}S_{\left[a\right|}\mathcal{L}_{0}D_{\left|b\right]}\mathsf{N}+P_{\left[ab\right]}+Q_{\left[ab\right]}=8\pi\mathcal{L}_{0}\mathcal{T}_{\left[ab\right]}\,,
		\end{aligned}
		\\
		&\begin{aligned}\mathcal{L}_{0}h_{ab} & =2\mathsf{N}K_{\left(ab\right)}+4\mathsf{N}S_{0\left(ab\right)}\,,
		\end{aligned}
		\label{eq_hyperbolic_system:vanishing_Saii_torsion_components}
	\end{aligned}
\end{equation}
form a non-strictly hyperbolic system for the variables $h_{ab}$, $K_{\left(ab\right)}$,
$S_{ab0}$. If the Cauchy data in a slice $\mathcal{V}$
is in the Gevrey class of index 2, the Cauchy problem has a unique
solution in a neighborhood of $\mathcal{V}$, with domain
of dependence determined by the isotropic cone of the metric $g$.
\end{prop}
\begin{proof}
The principal operator of the system~(\ref{eq_hyperbolic_system:vanishing_Saii_torsion_components})
can be written as a triangular matrix with diagonal elements either
$\square_{h}$, $\mathcal{L}_{0}^{2}$ or $\mathcal{L}_{0}$, where
we have omitted multiplicative $\mathsf{N}$ factors. Then, choosing
the value 2 for the Leray-Volevic weights for all the unknowns, and
the following weights for the equations
\begin{equation}
\begin{aligned}\begin{aligned}n[\square_{h}K_{\left(ab\right)}] & =n[\mathcal{L}_{0}^{2}S_{ab0}]=0\,, & n[\mathcal{L}_{0}h_{ab}] & =1\,,\end{aligned}
\end{aligned}
\end{equation}
the equations form a quasilinear system of differential equations.
The characteristic determinant for each pair $\left(ab\right)$ is
$\left(-\mathsf{N}^{-2}\xi_{0}^{2}+h^{ij}\xi_{i}\xi_{j}\right)\left(\xi_{0}\right)^{3}$.
The rest of the proof follows as in Proposition \ref{theorem:hyperbolic_system_general_torsion}. 
\end{proof}
\subsection{Geometric well-posedness for the Einstein-Cartan system}
\label{final-result}

The Einstein-Cartan equations, as the Einstein equations, were composed into constraint and evolution equations. In section \ref{sec:Propagation_contraints} we have proved that if the constraints are satisfied initially on a Cauchy hypersurface then they are satisfied in a future open neighborhood (see Proposition \ref{theorem:propagation_constraints_general}). Regarding the evolution equations, we have split the problem in the three cases of the previous subsections \ref{5.1}, \ref{5.2} and \ref{5.3}, which we were able to reduce to non-strictly hyperbolic systems of Leray-Ohya type to assert local well-posedness of the respective initial value problems. Gathering those previous results 
yields the following:
\begin{thm}
\label{theorem1}
Consider initial data given on Cauchy hypersurface satisfying the constraint equations~(\ref{eq:constraint_G00})--(\ref{eq:constraint_Gi0}), a stress-energy tensor
$\mathcal{T}_{\alpha\beta}$ verifying the conservation laws~(\ref{eq:conservation_law})
and a hypermomentum tensor $\Delta_{\alpha\beta\gamma}$ verifying
(\ref{eq:constraint_S_abc})--(\ref{eq:constraint_S_a00}). Under
the conditions either of Proposition~\ref{theorem:hyperbolic_system_general_torsion} or Proposition~\ref{prop2} or Proposition \ref{prop3},
the unique local solution $\left(h_{ab},K_{\left(ab\right)},S_{ab0},S_{ai}{}^{i}\right)$ of the respective quasilinear system  corresponds to a unique local solution of the Einstein-Cartan
field equations~(\ref{eq:efe1}) and (\ref{eq:efe2}), with domain of dependence determined by the light cone.
\end{thm}
We assumed that the stress-energy tensor was given {\em ab initio} and did not consider particular matter contents. When specifying the matter, it is possible that the system can be reduced to a Leray hyperbolic form~\cite{Leray_lecture_notes_1953} or even to a symmetric hyperbolic form, but that remains to be seen. Since both vacuum and minimally coupled scalar field Lagrangians are not compatible with a non-trivial torsion in the Einstein-Cartan theory, a naturally interesting matter field to investigate is electromagnetism and work in this direction is under progress.
%
\section*{Acknowlegments}

PL acknowledges partial financial support provided under the European Union\textquoteright s
H2020 ERC Advanced Grant \textquotedblleft Black holes: gravitational
engines of discovery\textquotedblright{} grant agreement no. Gravitas--101052587.
Views and opinions expressed are, however, those of the author only
and do not necessarily reflect those of the European Union or the
European Research Council. Neither the European Union nor the granting
authority can be held responsible for them.
FCM thanks H2020-MSCA-2022-SE project EinsteinWaves, GA No.101131233; CAMGSD, IST-ID, projects UIDB/04459/2020 and UIDP/04459/2020; CMAT, Univ. Minho, projects UIDB/00013/2020 and UIDP/00013/2020.
\appendix

\section{\label{Appendix:np1_identities} Geometric and $n+1$ identities}

\subsection{Identities for the derivatives of the lapse function and the shift
vector}

For the derivatives of the lapse function and the shift vector, we
find the following identities
\begin{equation}
\begin{aligned}\mathcal{L}_{\beta}\beta^{a} & =0\,,\\
\frac{1}{\mathsf{N}}\mathcal{L}_{0}\beta^{a} & =n^{\mu}\partial_{\mu}\beta^{a}+\beta^{i}\varphi^{a}{}_{i}-\beta^{i}K_{i}{}^{a}-2\beta^{i}S_{0i}{}^{a}\,;
\end{aligned}
\end{equation}
and~
\begin{equation}
\begin{aligned}\mathcal{L}_{\beta}\left(D_{a}\mathsf{N}\right) & =\beta^{i}D_{i}D_{a}\mathsf{N}+\left(D_{i}\mathsf{N}\right)\left(D_{a}\beta^{i}\right)+2\beta^{i}S_{ia}{}^{j}\left(D_{j}\mathsf{N}\right)\,,\\
\frac{1}{\mathsf{N}}\mathcal{L}_{0}\left(D_{a}\mathsf{N}\right) & =n^{\mu}\partial_{\mu}\left(D_{a}\mathsf{N}\right)+\left(D_{i}\mathsf{N}\right)K_{a}{}^{i}-\left(D_{i}\mathsf{N}\right)\varphi^{i}{}_{a}+2S_{0a}{}^{j}\left(D_{j}\mathsf{N}\right)\,.
\end{aligned}
\end{equation}
Moreover, we have the useful commutation relation
\begin{equation}
\begin{aligned}D_{a}\left(\mathcal{L}_{\beta}\mathsf{N}\right) & =\mathcal{L}_{\beta}\left(D_{a}\mathsf{N}\right)\,.\end{aligned}
\end{equation}
and assuming $\mathsf{N}$ to be a $\mathcal{C}^{2}$ function of the spacetime coordinates, we also have
\begin{equation}
D_{a}\left(\mathcal{L}_{0}\mathsf{N}\right)=\mathcal{L}_{0}\left(D_{a}\mathsf{N}\right)\,.
\end{equation}
%
\subsection{Identities for the derivatives of the first and second fundamental
form}

In the body of the text, we have the propagation equation~(\ref{eq:evolution_h}).
For the inverse induced metric, with components $h^{ab}$, we find~
\begin{equation}
\begin{aligned}\mathcal{L}_{\beta}h^{ab} & =-2D^{\left(a\right.}\beta^{\left.b\right)}-4S_{i}{}^{\left(ab\right)}\beta^{i}\,,\\
\mathcal{L}_{0}h^{ab} & =-2\mathsf{N}K^{\left(ab\right)}-4\mathsf{N}S_{0}{}^{\left(ab\right)}\,;
\end{aligned}
\end{equation}
whereas, for the second fundamental form we have~
\begin{equation}
\begin{aligned}\mathcal{L}_{\beta}K_{ab} & =\beta^{i}D_{i}K_{ab}+K_{ib}D_{a}\beta^{i}+K_{ai}D_{b}\beta^{i} +2\beta^{i}S_{ia}{}^{j}K_{jb}+2\beta^{i}S_{ib}{}^{j}K_{aj}\,,\\
\frac{1}{\mathsf{N}}\mathcal{L}_{0}K_{ab} & =n^{\mu}\partial_{\mu}K_{ab}+2S_{0a}{}^{i}K_{ib}+2S_{0b}{}^{i}K_{ai}
  +K_{ai}K^{i}{}_{b}+K_{ai}K_{b}{}^{i}-K_{ai}\varphi^{i}{}_{b}-K_{ib}\varphi^{i}{}_{a}\,,
\end{aligned}
\end{equation}
where $\varphi^{a}{}_{b}$ is defined in Eq.~(\ref{eq:definition_a_varphi}).

\subsection{Identities for the derivatives of the torsion field}

It is useful to find the derivatives associated with the $n+1$ decomposition
of the various components of the torsion tensor.

\paragraph{Derivatives of the $S_{a}$ components}

~
\begin{equation}
\begin{aligned}\mathcal{L}_{\beta}S_{a} & =\beta^{i}D_{i}S_{a}+S_{i}D_{a}\beta^{i}+2S_{ia}{}^{j}\beta^{i}S_{j}\,,\\
\frac{1}{\mathsf{N}}\mathcal{L}_{0}S_{a} & =n^{\mu}\partial_{\mu}S_{a}-S_{i}\varphi^{i}{}_{a}+2S_{0a}{}^{i}S_{i}+S_{i}K_{a}{}^{i}\,.
\end{aligned}
\end{equation}

\paragraph{Derivatives of the $S_{ab0}$ components}

~
\begin{equation}
\begin{aligned}\mathcal{L}_{\beta}S_{ab0} & =\beta^{i}D_{i}S_{ab0}+S_{ib0}D_{a}\beta^{i}+S_{ai0}D_{b}\beta^{i} +2\beta^{i}S_{ia}{}^{j}S_{jb0}+2\beta^{i}S_{ib}{}^{j}S_{aj0}\,,\\
\frac{1}{\mathsf{N}}\mathcal{L}_{0}S_{ab0} & =n^{\mu}\partial_{\mu}S_{ab0}+2S_{0a}{}^{i}S_{ib0}+2S_{0b}{}^{i}S_{ai0} -S_{ai0}\varphi^{i}{}_{b}-S_{ib0}\varphi^{i}{}_{a}+S_{ai0}K_{b}{}^{i}+S_{ib0}K_{a}{}^{i}\,.
\end{aligned}
\end{equation}

\paragraph{Derivatives of the $S_{0ab}$ components}

~
\begin{equation}
\begin{aligned}\mathcal{L}_{\beta}S_{0ab} & =\beta^{i}D_{i}S_{0ab}+S_{0ib}D_{a}\beta^{i}+S_{0ai}D_{b}\beta^{i} +2\beta^{i}S_{ia}{}^{j}S_{0jb}+2\beta^{i}S_{ib}{}^{j}S_{0aj}\,,\\
\frac{1}{\mathsf{N}}\mathcal{L}_{0}S_{0ab} & =n^{\mu}\partial_{\mu}S_{0ab}+2S_{0a}{}^{i}S_{0ib}+2S_{0b}{}^{i}S_{0ai} -S_{0ai}\varphi^{i}{}_{b}-S_{0ib}\varphi^{i}{}_{a}+S_{0ai}K_{b}{}^{i}+S_{0ib}K_{a}{}^{i}\,.
\end{aligned}
\end{equation}

\paragraph{Derivatives of the $S_{abc}$ component}

~~~
\begin{equation}
\begin{aligned}
	\mathcal{L}_{\beta}S_{abc} & =\beta^{i}D_{i}S_{abc}+S_{ibc}D_{a}\beta^{i}+S_{aic}D_{b}\beta^{i}+S_{abi}D_{c}\beta^{i}+2\beta^{i}\left(S_{ia}{}^{j}S_{jbc}+S_{ib}{}^{j}S_{ajc}+S_{ic}{}^{j}S_{abj}\right)\\
\frac{1}{\mathsf{N}}\mathcal{L}_{0}S_{abc} & =n^{\mu}\partial_{\mu}S_{abc}-S_{ibc}\varphi^{i}{}_{a}-S_{aic}\varphi^{i}{}_{b}-S_{abi}\varphi^{i}{}_{c} +2S_{0a}{}^{i}S_{ibc}+2S_{0b}{}^{i}S_{aic}+2S_{0c}{}^{i}S_{abi}\\
 & +S_{ibc}K_{a}{}^{i}+S_{aic}K_{b}{}^{i}+S_{abi}K_{c}{}^{i}\,.
\end{aligned}
\end{equation}

\subsection{Ricci tensor and Ricci scalar}

A direct application of the Gauss-Codazzi-Ricci embedding equations
gives the relation between the intrinsic Ricci
scalar of an embedded hypersurface, $\ps{\left(N-1\right)}R$,
and the Ricci scalar, $R$, of the embedding manifold. From the definition $R_{\alpha\gamma}:=R_{\alpha\mu\gamma}{}^{\mu}$
and properties of the projector tensor, with components $h_{\alpha\beta}$,
we find
\begin{equation}
R_{\alpha\gamma}=h^{bd}e_{b}^{\beta}e_{d}^{\delta}R_{\alpha\beta\gamma\delta}+\varepsilon R_{\alpha\beta\gamma\delta}n^{\beta}n^{\delta}.
\label{Appendix_eq:Ricci_push-forward_Riemann}
\end{equation}
This equation together with the field equations of the Einstein-Cartan theory imply Eq.~(\ref{eq:propagation_Gij}).

It is also useful to relate the full Ricci tensor,
$R_{\alpha\beta}$, and the metric Ricci tensor $\tilde{R}_{\alpha\beta}$, such that
\begin{equation}
R_{\alpha\beta}=\tilde{R}_{\alpha\beta}+\nabla_{\gamma}\mathcal{K}_{\alpha\beta}{}^{\gamma}-\nabla_{\alpha}\mathcal{K}_{\gamma\beta}{}^{\gamma}+2S_{\sigma\alpha}{}^{\gamma}\mathcal{K}_{\gamma\beta}{}^{\sigma}+\mathcal{K}_{\gamma\beta}{}^{\sigma}\mathcal{K}_{\alpha\sigma}{}^{\gamma}-\mathcal{K}_{\gamma\sigma}{}^{\gamma}\mathcal{K}_{\alpha\beta}{}^{\sigma}\,,\label{Appendix_eq:Full_Ricci_metric_Ricci_relation}
\end{equation}
where $\mathcal{K}_{\alpha\beta\gamma}$ represent the components of the controsion
tensor, defined in Eq.~(\ref{eq:Contorsion}). Taking the trace of
Eq.~(\ref{Appendix_eq:Ricci_push-forward_Riemann}), yields
\begin{equation}
R=h^{ac}h^{bd}e_{a}^{\alpha}e_{b}^{\beta}e_{c}^{\gamma}e_{d}^{\delta}R_{\alpha\beta\gamma\delta}+2\varepsilon h^{ac}e_{a}^{\alpha}n^{\beta}e_{c}^{\gamma}n^{\delta}R_{\alpha\beta\gamma\delta}\,.
\end{equation}
Considering the Gauss and Ricci equations, Eqs.~(\ref{eq:Gauss_equation})
and (\ref{eq:Ricci_equation}), respectively, yields
\begin{equation}
\begin{aligned}R=  \ps{\left(N-1\right)}R&-2a_{i}a^{i}+4S_{i}a^{i}-\varepsilon (K^{2}+K_{ij}K^{ji})\\ &+2\varepsilon\left(D_{i}a^{i}-h^{ac}\dot{K}_{ac}+2K_{ij}\varphi^{(ji)}-2S_{0ij}K^{ji}\right)\,,
\end{aligned}
\end{equation}
or
\begin{equation}
\begin{aligned}R= & \ps{\left(N-1\right)}R+2\varepsilon D_{i}a^{i}-2a_{i}a^{i}+4S_{i}a^{i} -\frac{2}{\mathsf{N}}\varepsilon\mathcal{L}_{0}K-4\varepsilon K_{ij}S_{0}{}^{ji}-\varepsilon K_{ij}K^{ji}-\varepsilon K^{2}\,,
\end{aligned}
\end{equation}
where $K:=K_{i}{}^{i}$.
\section{\label{Appendix:conservation_laws} Conservation equations}

In the $n+1$ formalism, the conservation equations \eqref{eq:conservation_law} read
\begin{equation}
\begin{aligned}
\frac{\varepsilon}{\mathsf{N}}\mathcal{L}_{0}\mathcal{T}_{00}&+D_{i}\mathcal{T}_{0}{}^{i}-\frac{\varepsilon}{4\pi\mathsf{N}}S^{0ba}\mathcal{L}_{0}K_{ab}+\frac{1}{4\pi\mathsf{N}}S^{abc}\left(\frac{1}{2}\mathcal{L}_{0}S_{abc}+\mathcal{L}_{0}S_{acb}\right)\\&=  \frac{1}{4\pi}S^{abc}\left\{ \frac{1}{2}\left(D_{a}K_{bc}+D_{a}K_{cb}+D_{c}K_{ba}\right)+D_{a}S_{0bc}+D_{a}S_{0cb}+D_{c}S_{0ba}\right.\\
 & +\frac{1}{\mathsf{N}}\left(S_{0cb}D_{a}\mathsf{N}+S_{a0b}D_{c}\mathsf{N}+S_{ac0}D_{b}\mathsf{N}+S_{0bc}D_{a}\mathsf{N}+\frac{1}{2}S_{ab0}D_{c}\mathsf{N}\right)\\
 & +2S_{i0a}S_{bc}{}^{i}+S_{i0c}S_{ba}{}^{i}+S_{c0i}S_{ba}{}^{i}+2S_{b0}{}^{i}S_{cai}+S_{aci}K_{b}{}^{i}\\
 & \left.+\frac{1}{2}S_{abi}K_{c}{}^{i}+\frac{1}{2}\left(S_{ab}{}^{i}K_{ic}+S_{cb}{}^{i}K_{ia}+S_{ac}{}^{i}K_{ib}\right)+2\varepsilon S_{a}K_{cb}\right\} \\
 & -\frac{1}{4\pi}S^{0ba}\left(\varepsilon K_{ai}K_{b}{}^{i}+\varepsilon D_{a}a_{b}-a_{a}a_{b}+2S_{a}a_{b}+2\varepsilon S_{0b}{}^{i}K_{ai}\right)\\
 & +2S_{0ji}\mathcal{T}^{ij}-2\varepsilon S_{i}\mathcal{T}^{0i}+\mathcal{T}_{ij}K^{ji}-\varepsilon\mathcal{T}_{00}K+\varepsilon\left(\mathcal{T}_{i0}+\mathcal{T}_{0i}\right)a^{i}\\
 & -\frac{1}{4\pi\left(N-2\right)}S_{0i}{}^{i}\left(8\pi\mathcal{T}_{j}{}^{j}+8\pi\varepsilon\mathcal{T}_{00}-2\Lambda\right)
\end{aligned}
\label{eq:propagation_L0_T00}
\end{equation}
\begin{equation}
\begin{aligned}\frac{\varepsilon}{\mathsf{N}}\mathcal{L}_{0}\mathcal{T}_{a0}+&D_{i}\mathcal{T}_{a}{}^{i}-\frac{1}{4\pi\mathsf{N}}S^{i}\mathcal{L}_{0}K_{ai}
-\frac{\varepsilon}{4\pi\mathsf{N}}S^{ij0}\left(\frac{1}{2}\mathcal{L}_{0}S_{ija}-\mathcal{L}_{0}S_{aij}\right)\\
= & -\frac{\varepsilon}{4\pi}S^{jb0}\left\{ D_{j}K_{ab}+2\varepsilon S_{j}K_{ab}-2S_{0ij}S_{ba}{}^{i}-2S_{0b}{}^{i}S_{aji}\right.\\
 & +\frac{1}{\mathsf{N}}\left[2D_{j}\left(\mathsf{N}S_{0\left(ab\right)}\right)+D_{a}\left(\mathsf{N}S_{0bj}\right)+D_{j}\left(\mathsf{N}S_{ab0}\right)+\frac{1}{2}D_{a}\left(\mathsf{N}S_{jb0}\right)\right]\\
 & \left.+S_{jb}{}^{i}K_{\left(ia\right)}+2S_{ja}{}^{i}K_{\left(ib\right)}-2S_{0\left(ia\right)}S_{bj}{}^{i}\right\} \\
 & -\frac{1}{8\pi}S^{cdb}\left[\ps{\left(N-1\right)}{R_{abcd}}-2\varepsilon K_{ac}K_{bd}\right]+\frac{\varepsilon}{2\pi}S^{0cb}\left(D_{\left[a\right.}K_{\left.b\right]c}+S_{ab}{}^{i}K_{ic}\right)\\
 & +\frac{2}{N-2}\left(S_{ai}{}^{i}+\varepsilon S_{a}\right)\left(\frac{\Lambda}{4\pi}-\mathcal{T}_{j}{}^{j}-\varepsilon\mathcal{T}_{00}\right)+2\varepsilon S_{ai0}\mathcal{T}^{0i}+\varepsilon\mathcal{T}_{a}{}^{i}a_{i}\\
 & +\frac{\varepsilon}{\mathsf{N}}\mathcal{T}_{00}D_{a}\mathsf{N}-\varepsilon\left(\mathcal{T}_{0i}K^{i}{}_{a}-\mathcal{T}_{i0}K_{a}{}^{i}+\mathcal{T}_{a0}K\right)+2S_{aij}\mathcal{T}^{ji}\\
 & -\frac{1}{4\pi}S^{b}\left(K_{ai}K_{b}{}^{i}+D_{a}a_{b}-\varepsilon a_{a}a_{b}+2\varepsilon S_{a}a_{b}+2S_{0b}{}^{i}K_{ai}\right).
\end{aligned}
\label{eq:propagation_L0_Ti0}
\end{equation}
\section{\label{sec:Wave-equation-for-K_(ab)} Wave equation for $K_{\left(ab\right)}$ and
the densitized lapse}

From Eqs.~(\ref{eq:Connection}) and (\ref{eq:Christoffel_symbols})
and $\delta g^{\lambda\sigma}=-g^{\mu\lambda}g^{\nu\sigma}\left(\delta g_{\mu\nu}\right)$,
where $\delta g_{\mu\nu}$ represents the component $\mu\nu$ of a
infinitesimal variation of the metric tensor, it is straightforward
to show that, as expected, an infinitesimal variation of the connection
$\delta C_{\alpha\beta}^{\gamma}$, is a tensor field. This, in conjunction
with Eq.~(\ref{eq:Ricci_tensor_connnection}), allows us to derive
the generalized Palatini identity for a general metric compatible
connection:
\begin{equation}
\delta R_{\alpha\beta}=\nabla_{\mu}\left(\delta C_{\alpha\beta}^{\mu}\right)-\nabla_{\alpha}\left(\delta C_{\mu\beta}^{\mu}\right)+2S_{\nu\alpha}{}^{\mu}\left(\delta C_{\mu\beta}^{\nu}\right)\,,
\end{equation}
relating an infinitesimal variation of the Ricci tensor components
with an infinitesimal variation of the connection coefficients. Using
the generalized Palatini identity and finding an expression for an
infinitesimal variation of the connection coefficients in terms of
a variation of the metric $g$ and the variation of the torsion $S$,
we find the useful formula
\begin{equation}
\begin{aligned}\delta R_{\alpha\beta} & =\frac{1}{2}\left[g^{\gamma\sigma}\nabla_{\gamma}\left(\nabla_{\alpha}\delta g_{\sigma\beta}+\nabla_{\beta}\delta g_{\alpha\sigma}\right)-\nabla_{\sigma}\nabla^{\sigma}\delta g_{\alpha\beta}-\nabla_{\alpha}\nabla_{\beta}\left(g^{\gamma\sigma}\delta g_{\gamma\sigma}\right)\right]\\
 & +\nabla_{\gamma}\left[2S_{\left(\alpha\right|}{}^{\gamma\mu}\delta g_{\left|\beta\right)\mu}+2g^{\gamma\sigma}S^{\mu}{}_{\left(\alpha\beta\right)}\delta g_{\mu\sigma}+\delta \mathcal{K}_{\alpha\beta}{}^{\gamma}\right]-2\nabla_{\alpha}\left(\delta S_{\mu\beta}{}^{\mu}\right)\\
 & +2S_{\gamma\alpha}{}^{\mu}\left[\frac{1}{2}g^{\gamma\sigma}\left(\nabla_{\mu}\delta g_{\sigma\beta}+\nabla_{\beta}\delta g_{\mu\sigma}-\nabla_{\sigma}\delta g_{\mu\beta}\right)+S_{\mu}{}^{\gamma\nu}\delta g_{\nu\beta}\right.\\
 & +S_{\beta}{}^{\gamma\nu}\delta g_{\mu\nu}+2g^{\gamma\sigma}S^{\nu}{}_{\left(\mu\beta\right)}\delta g_{\nu\sigma}+\delta \mathcal{K}_{\mu\beta}{}^{\gamma}\Bigr]\,,
\end{aligned}
\label{eq_Wave:Variation_Ricci_metric_torsion}
\end{equation}
where the infinitesimal variation of the contorsion tensor, $\delta \mathcal{K}_{\alpha\beta}{}^{\gamma}$,
is given in terms of an infinitesimal variation of the torsion tensor
by 
\begin{equation}
\delta \mathcal{K}_{\alpha\beta}{}^{\gamma}=\delta S_{\alpha\beta}{}^{\gamma}+\delta S^{\gamma}{}_{\alpha\beta}-\delta S_{\beta}{}^{\gamma}{}_{\alpha}.
\end{equation}
Equation~(\ref{eq_Wave:Variation_Ricci_metric_torsion}) together
with Eq.~(\ref{eq:evolution_h}) allow us to find the following expression
for the Lie derivative of the intrinsic Ricci tensor $\ps{\left(N-1\right)}{R_{ab}}$
along the vector field $t-\beta$ as
\begin{equation}
\begin{aligned}\mathcal{L}_{0}\ps{\left(N-1\right)}{R_{ab}} & =D^{i}D_{a}\left[\mathsf{N}K_{\left(ib\right)}+2\mathsf{N}S_{0\left(ib\right)}\right]+D^{i}D_{b}\left[\mathsf{N}K_{\left(ia\right)}+2\mathsf{N}S_{0\left(ia\right)}\right]\\
 & -D^{i}D_{i}\left[\mathsf{N}K_{\left(ab\right)}+2\mathsf{N}S_{0\left(ab\right)}\right]-D_{a}D_{b}\left[\mathsf{N}K+2\mathsf{N}S_{0i}{}^{i}\right]\\
 & +D_{i}\left(\mathcal{L}_{0}K_{ab}{}^{i}\right)-2D_{a}\left(\mathcal{L}_{0}S_{ib}{}^{i}\right)+M_{ab}\,,
\end{aligned}
\label{eq_Wave:L0_R_ab_intrinsic}
\end{equation}
where
\begin{equation}
\begin{aligned}M_{ab}= & 2S_{a}{}^{kj}\left[D_{k}\left(\mathsf{N}K_{\left(jb\right)}+2\mathsf{N}S_{0\left(jb\right)}\right)-D_{j}\left(\mathsf{N}K_{\left(kb\right)}+2\mathsf{N}S_{0\left(kb\right)}\right)-D_{b}\left(\mathsf{N}K_{\left(jk\right)}+2\mathsf{N}S_{0\left(jk\right)}\right)\right]\\
 & +4S_{ia}{}^{j}\left[S_{j}{}^{ik}\left(\mathsf{N}K_{\left(kb\right)}+2\mathsf{N}S_{0\left(kb\right)}\right)+S_{b}{}^{ik}\left(\mathsf{N}K_{\left(kj\right)}+2\mathsf{N}S_{0\left(kj\right)}\right)\right]+2S_{ia}{}^{j}\left(\mathcal{L}_{0}K_{jb}{}^{i}\right)\\
 & +4D^{j}\left[S^{k}{}_{\left(ab\right)}\left(\mathsf{N}K_{\left(kj\right)}+2\mathsf{N}S_{0\left(kj\right)}\right)\right]-8S_{a}{}^{kj}S^{l}{}_{\left(jb\right)}\left(\mathsf{N}K_{\left(lk\right)}+2\mathsf{N}S_{0\left(lk\right)}\right)\\
 & +2D_{i}\left[S_{a}{}^{ij}\left(\mathsf{N}K_{\left(jb\right)}+2\mathsf{N}S_{0\left(jb\right)}\right)+S_{b}{}^{ij}\left(\mathsf{N}K_{\left(ja\right)}+2\mathsf{N}S_{0\left(ja\right)}\right)\right]
\end{aligned}
\label{eq_Wave:M_tensor_def}
\end{equation}
contains only terms of at most first-order derivatives of $K_{ab}$,
$S_{abc}$, $S_{0ab}$ and $\mathsf{N}$.

Lastly, we will also need to compute the relation between the quantities
$\mathcal{L}_{0}\left(D_{a}\partial_{b}\mathsf{N}\right)$ and $D_{a}\partial_{b}\left(\mathcal{L}_{0}\mathsf{N}\right)$.
For this, it is useful to rewrite the second Codazzi equation, Eq.~\eqref{eq:Second_Codazzi_equation}, explicitly
in terms of the lapse function, the second fundamental form, and the
induced torsion tensor, such that
\begin{equation}
\begin{aligned}\frac{1}{2}e_{a}^{\alpha}e_{b}^{\beta}e_{c}^{\gamma}n^{\delta}\left(R_{\gamma\delta\alpha\beta}-R_{\alpha\beta\gamma\delta}\right) & =-\frac{1}{2N}\mathcal{L}_{0}S_{abc}+\frac{1}{N}\mathcal{L}_{0}S_{c\left[ab\right]}+\frac{1}{N}D_{\left[a\right|}\left(NS_{0\left|b\right]c}\right)+\frac{1}{N}D_{\left[a\right|}\left(NS_{0c\left|b\right]}\right)\\
 & +\frac{1}{N}D_{c}\left(NS_{0\left[ba\right]}\right)-\frac{3}{2}D_{\left[a\right.}K_{\left.bc\right]}+\frac{1}{2N}\left(S_{ab0}D_{c}N-2S_{c\left[a\right|0}D_{\left|b\right]}N\right)\\
 & -2S_{0i\left[a\right.}S_{\left.b\right]c}{}^{i}-2S_{c\left[a\right|i}S_{0\left|b\right]}{}^{i}-2S_{ba}{}^{i}S_{0\left(ci\right)}-\frac{3}{2}S_{\left[ab\right|}{}^{i}K_{i\left|c\right]}\\
 & -\frac{1}{2}S_{bai}K_{c}{}^{i}-S_{c\left[a\right|i}K_{\left|b\right]}{}^{i}+2\varepsilon S_{\left[a\right|}K_{c\left|b\right]}\,.
\end{aligned}
\end{equation}
Then, in conjunction with Eqs.~(\ref{eq:Riemann_tensor_definition})
and (\ref{eq:first_Bianchi_identity}), after a lengthy and laborious
calculation, we find
\begin{equation}
\begin{aligned}\mathcal{L}_{0}\left(D_{a}\partial_{b}\mathsf{N}\right) & =\frac{1}{2}\left(\partial^{c}\mathsf{N}\right)\left[4\mathcal{L}_{0}S_{a\left[cb\right]}-2\mathcal{L}_{0}S_{cba}-D_{a}\left(\mathcal{L}_{0}h_{bc}\right)-D_{b}\left(\mathcal{L}_{0}h_{ac}\right)+D_{c}\left(\mathcal{L}_{0}h_{ab}\right)\right]\\
 & +2N\left(\partial^{c}\mathsf{N}\right)\left[S_{ca}{}^{i}K_{\left(ib\right)}+S_{cb}{}^{i}K_{\left(ai\right)}+S_{ab}{}^{i}K_{\left(ic\right)}+2S_{ca}{}^{i}S_{0\left(bi\right)}\right.\\
 & \left.+2S_{ab}{}^{i}S_{0\left(ci\right)}+2S_{cb}{}^{i}S_{0\left(ai\right)}\right]+D_{a}\partial_{b}\left(\mathcal{L}_{0}\mathsf{N}\right)\,.
\end{aligned}
\end{equation}

We are now in a position to find a wave equation for the second fundamental
form, or more precisely, for its symmetric part, $K_{\left(ab\right)}$.
In order to simplify the notation, we will introduce $R_{ab}:=e_{a}^{\alpha}e_{b}^{\beta}R_{\alpha\beta}$,
for the pull-back of the Ricci tensor in $\mathcal{M}$ (to not be
confused with the intrinsic Ricci tensor components $\ps{\left(N-1\right)}{R_{ab}}$).
Gathering the previous results in this subsection, and considering
Eqs~(\ref{eq:constraint_Gi0}), (\ref{eq:propagation_Gij}), (\ref{eq:evolution_h})
and (\ref{eq:constraint_a_N}) we find
\begin{equation}
\begin{aligned}\mathcal{L}_{0}R_{ab}-2D_{\left(a\right|}\left(\mathsf{N}R_{\left|b\right)0}\right) & =-\mathsf{N}\square_{h}K_{\left(ab\right)}+\frac{1}{\mathsf{N}}\mathcal{L}_{0}^{2}K_{\left[ab\right]}-\frac{1}{\mathsf{N}}D_{a}D_{b}\left[\mathcal{L}_{0}\mathsf{N}-\mathsf{N}^{2}K-2\mathsf{N}^{2}S_{0i}{}^{i}\right]\\
 & +2D^{j}D_{\left(a\right|}\left(\mathsf{N}S_{0\left|b\right)j}\right)+2D^{j}D_{\left(a\right|}\left(\mathsf{N}S_{\left|b\right)j0}\right)-2D^{i}D_{i}\left(\mathsf{N}S_{0\left(ab\right)}\right)\\
 & -4D_{a}D_{b}\left(\mathsf{N}S_{0i}{}^{i}\right)+2D^{j}D_{\left(a\right|}\left(\mathsf{N}S_{0j\left|b\right)}\right)-2\mathcal{L}_{0}D_{a}S_{b}\\
 & -\frac{2}{\mathsf{N}}S_{b}\mathcal{L}_{0}D_{a}\mathsf{N}+D_{i}\mathcal{L}_{0}K_{ab}{}^{i}-2D_{a}\mathcal{L}_{0}S_{ib}{}^{i}\\
 & +\frac{1}{\mathsf{N}^{2}}\left[\mathcal{L}_{0}\mathsf{N}-\mathsf{N}^{2}K-2\mathsf{N}^{2}S_{0i}{}^{i}\right]D_{a}D_{b}\mathsf{N} +P_{ab}+Q_{ab}\,,
\end{aligned}
\label{eq_Wave:Wave_eq_general}
\end{equation}
where we introduce the hyperbolic operator $\square_{h}:=-\frac{1}{\mathsf{N}^{2}}\mathcal{L}_{0}^{2}+\tilde{D}_{i}\tilde{D}^{i}$ and
$\tilde{D}$ represents the Levi-Civita connection associated with
$h_{ab}$, and the tensors
\begin{equation}
\begin{aligned}
	\label{P_tensor}
	P_{ab}= & -K_{\left(ab\right)}\left(D_{i}D^{i}\mathsf{N}\right)-2\left(D^{i}\mathsf{N}\right)\left(D_{i}K_{\left(ab\right)}\right)+2D_{\left(a\right|}\left[K_{\left|b\right)}{}^{i}D_{i}\mathsf{N}-KD_{\left|b\right)}\mathsf{N}\right]\\
 & -\frac{2}{\mathsf{N}}\left(D_{\left(a\right|}\mathsf{N}\right)D_{\left|b\right)}\left(\mathsf{N}K\right)+2\mathsf{N}\left[\ps{\left(N-1\right)}{R_{i\left(ab\right)}{}^{j}}K_{j}{}^{i}+\ps{\left(N-1\right)}{R_{\left(a\right|}{}^{j}}K_{\left|b\right)j}\right]\\
 & +\frac{1}{2\mathsf{N}}\left(\partial^{c}\mathsf{N}\right)\left[D_{a}\left(\mathcal{L}_{0}h_{bc}\right)+D_{b}\left(\mathcal{L}_{0}h_{ac}\right)-D_{c}\left(\mathcal{L}_{0}h_{ab}\right)\right]\\
 & +\mathcal{L}_{0}\left(K_{ab}K-2K_{a}{}^{i}K_{\left(ib\right)}\right)-\frac{1}{\mathsf{N}^{2}}\left(\mathcal{L}_{0}\mathsf{N}\right)\left(\mathcal{L}_{0}K_{ab}\right)\,,
\end{aligned}
\end{equation}
and
\begin{equation}
\begin{aligned}
	\label{Q_tensor}
Q_{ab}= & D_{i}\left[S_{a}{}^{ij}\mathcal{L}_{0}h_{jb}+S_{b}{}^{ij}\mathcal{L}_{0}h_{ja}+2h^{ij}S^{k}{}_{\left(ab\right)}\mathcal{L}_{0}h_{kj}\right]\\
&-S_{a}{}^{ij}\left[D_{j}\mathcal{L}_{0}h_{ib}-D_{i}\mathcal{L}_{0}h_{jb}+D_{b}\mathcal{L}_{0}h_{ij}\right]\\
 & -\frac{1}{\mathsf{N}}\left(\partial^{c}\mathsf{N}\right)\left(\mathcal{L}_{0}S_{cba}+2\mathcal{L}_{0}S_{a\left[cb\right]}+S_{ca}{}^{i}\mathcal{L}_{0}h_{bi}+S_{ab}{}^{i}\mathcal{L}_{0}h_{ci}+S_{cb}{}^{i}\mathcal{L}_{0}h_{ai}\right)\\
 & -\mathsf{N}\mathcal{K}^{ij}{}_{\left(a\right|}D_{i}\left[K_{j\left|b\right)}+K_{\left|b\right)j}\right]+\mathsf{N}\tilde{D}^{i}\left[\mathcal{K}_{i\left(a\right|}{}^{j}K_{j\left|b\right)}+\mathcal{K}_{i\left(a\right|}{}^{j}K_{\left|b\right)j}\right]\\
 & +2S_{ia}{}^{j}\left[S_{j}{}^{ik}\mathcal{L}_{0}h_{kb}+S_{b}{}^{ik}\mathcal{L}_{0}h_{kj}+2h^{ik}S^{l}{}_{\left(jb\right)}\mathcal{L}_{0}h_{lk}+\mathcal{L}_{0}\mathcal{K}_{jb}{}^{i}\right]\\
 & -\left(D_{a}\mathsf{N}\right)\mathcal{L}_{0}\left(\frac{2}{\mathsf{N}}S_{b}\right)+2S_{ab}{}^{i}D_{i}\left(\mathsf{N}K\right)-2\mathcal{L}_{0}\left(K_{ai}S_{0b}{}^{i}\right)\\
 & +2\mathsf{N}S^{ji}{}_{i}D_{j}K_{\left(ab\right)}-\frac{4}{\mathsf{N}}\left(D_{\left(a\right|}\mathsf{N}\right)D_{\left|b\right)}\left(\mathsf{N}S_{0i}{}^{i}\right)\\
 & +4D_{\left(a\right|}\left(\mathsf{N}S_{\left|b\right)ij}K^{ji}\right)+4S_{\left(a\right|}{}^{ji}D_{i}\left(\mathsf{N}K_{\left|b\right)j}\right)\,,
\end{aligned}
\end{equation}
contain only terms with derivatives of at most second-order in $h_{ab}$
and the lapse function $\mathsf{N}$ and with derivatives of at most
first-order of $K_{ab}$, $S_{abc}$, $S_{0ab}$ and $S_{a}$. Note
that $Q_{ab}$ vanishes if the torsion tensor is such that
the components $S_{abc}$, $S_{0ab}$ and $S_{a}$ are identically
zero.

Now, in (\ref{eq_Wave:Wave_eq_general}) the term $D_{a}D_{b}\left[\mathcal{L}_{0}\mathsf{N}-\mathsf{N}^{2}K-2\mathsf{N}^{2}S_{0i}{}^{i}\right]$
contains third-order derivatives of $\mathsf{N}$ and second-order
derivatives of $K$, spoiling the hyperbolicity of the equation. To
remove this term, we use the gauge freedom for the lapse
function and the shift vector. Setting the term in parenthesis to zero is equivalent to impose the algebraic gauge, Eq.~\eqref{gauge}, leading to the definition of a densitized lapse, Eq.~\eqref{eq_Wave:Algebraic_gauge}.
Then, substituting (\ref{eq_Wave:Algebraic_gauge}) in (\ref{eq_Wave:Wave_eq_general})
we find
\begin{equation}
\begin{aligned}\mathcal{L}_{0}R_{ab}-2D_{\left(a\right|}\left(\mathsf{N}R_{\left|b\right)0}\right) & =-\mathsf{N}\square_{h}K_{\left(ab\right)}+\frac{1}{\mathsf{N}}\mathcal{L}_{0}^{2}K_{\left[ab\right]}+2D^{j}D_{\left(a\right|}\left(\mathsf{N}S_{0\left|b\right)j}\right)\\
 & +2D^{j}D_{\left(a\right|}\left(\mathsf{N}S_{\left|b\right)j0}\right)-2D^{i}D_{i}\left(\mathsf{N}S_{0\left(ab\right)}\right)+2D_{a}\mathcal{L}_{0}S_{bi}{}^{i}\\
 & -4D_{a}D_{b}\left(\mathsf{N}S_{0i}{}^{i}\right)+2D^{j}D_{\left(a\right|}\left(\mathsf{N}S_{0j\left|b\right)}\right)-2\mathcal{L}_{0}D_{a}S_{b}\\
 & -\frac{2}{\mathsf{N}}S_{b}\mathcal{L}_{0}D_{a}\mathsf{N}+D_{i}\mathcal{L}_{0}\mathcal{K}_{ab}{}^{i}+P_{ab}+Q_{ab}
\end{aligned}
\label{eq_Wave:Wave_eq_densitized_full}
\end{equation}
and separating it in its symmetric and anti-symmetric parts gives
\begin{equation}
\begin{aligned}\mathcal{L}_{0}R_{\left(ab\right)}-2D_{\left(a\right|}\left(\mathsf{N}R_{\left|b\right)0}\right) & =-\mathsf{N}\square_{h}K_{\left(ab\right)}+2D^{j}D_{\left(a\right|}\left(\mathsf{N}S_{0\left|b\right)j}\right)+2D^{j}D_{\left(a\right|}\left(\mathsf{N}S_{\left|b\right)j0}\right)\\
 & -4D_{\left(a\right.}D_{\left.b\right)}\left(\mathsf{N}S_{0i}{}^{i}\right)+2D^{j}D_{\left(a\right|}\left(\mathsf{N}S_{0j\left|b\right)}\right)+2D_{i}\mathcal{L}_{0}S^{i}{}_{\left(ab\right)}\\
 & -2D^{i}D_{i}\left(\mathsf{N}S_{0\left(ab\right)}\right)+2D_{\left(a\right|}\mathcal{L}_{0}S_{\left|b\right)i}{}^{i}-2\mathcal{L}_{0}D_{\left(a\right|}S_{\left|b\right)}\\
 & -\frac{2}{\mathsf{N}}S_{\left(a\right|}\mathcal{L}_{0}D_{\left|b\right)}\mathsf{N}+P_{\left(ab\right)}+Q_{\left(ab\right)}\,,
\end{aligned}
\label{eq_Wave:Wave_eq_densitized_symmetric}
\end{equation}
\begin{equation}
\begin{aligned}\mathcal{L}_{0}R_{\left[ab\right]} & =-\frac{1}{\mathsf{N}}\mathcal{L}_{0}^{2}S_{ab0}+4S_{ab}{}^{j}D_{j}\left(\mathsf{N}S_{0i}{}^{i}\right)-2\mathcal{L}_{0}D_{\left[a\right.}S_{\left.b\right]}+D_{i}\mathcal{L}_{0}S_{ab}{}^{i}\\
 & +2D_{\left[a\right|}\mathcal{L}_{0}S_{\left|b\right]i}{}^{i}+\frac{2}{\mathsf{N}}S_{\left[a\right|}\mathcal{L}_{0}D_{\left|b\right]}\mathsf{N}+P_{\left[ab\right]}+Q_{\left[ab\right]}\,,
\end{aligned}
\label{eq_Wave:Wave_eq_densitized_antisymmetric}
\end{equation}
where we have used Eq.~(\ref{eq:antisymmetric_2nd_Fundamental_form_torsion_relation})
to write Eq.~(\ref{eq_Wave:Wave_eq_densitized_antisymmetric}) explicitly
as an evolution equation for the torsion component $S_{ab0}$.
%

\end{document}